\documentclass[a4paper,USenglish, 11pt]{article}
\usepackage{amsthm,amsmath,amssymb,amsfonts,hyperref,setspace}
\hypersetup{colorlinks=true,allcolors=blue}
\usepackage{color}
\usepackage{algorithm}
\usepackage{algorithmic} 
\usepackage{enumerate}

\usepackage{thmtools}
\usepackage{thm-restate}
\usepackage{xspace}

\usepackage{csquotes}
\MakeOuterQuote{"}

\DeclareMathOperator*{\argmin}{arg\,min}

\newtheorem{theorem}{Theorem}[section]

\newtheorem{lemma}[theorem]{Lemma}
\newtheorem{claim}[theorem]{Claim}
\newtheorem{corollary}[theorem]{Corollary}
\newtheorem{proposition}[theorem]{Proposition}
\newtheorem{observation}[theorem]{Observation}

\theoremstyle{remark}
\newtheorem{remark}{Remark}[section]

\usepackage{fullpage}
\newcommand{\comment}[1]{}

\newcommand{\sym}{\mathcal{S}}
\newcommand{\obj}{\texttt{Obj}}
\newcommand{\opt}{\texttt{OPT}}

\newcommand{\BI}{\textsc{BestFromInput}\xspace}
\newcommand{\BRO}{\textsc{RelativeOrder}\xspace}

\newcommand{\ed}{\Delta}
\newcommand{\lcs}{\textrm{lcs}}
\newcommand{\med}{x_\textrm{med}}

\floatname{algorithm}{Procedure}

\begin{document}

%\thispagestyle{empty}
%\pagenumbering{gobble}

%\newcommand*\samethanks[1][\value{footnote}]{\footnotemark[#1}
%\thispagestyle{empty}
%\setcounter{page}{0}

\title{Approximating the Median under the Ulam Metric} % 
%\title{Ulam Median: Breaking the 2-approximation}
%\title{Approximate Median over the Ulam Metric: Breaking 2-factor in Polytime}

\author{%
Diptarka Chakraborty%
\thanks{National University of Singapore.
  Work partially supported by NUS ODPRT Grant, WBS No. R-252-000-A94-133.
    Email: \texttt{diptarka@comp.nus.edu.sg}
  }
\and
Debarati Das%
\thanks{Basic Algorithm Research Copenhagen (BARC), University of Copenhagen
	Email: \texttt{debaratix710@gmail.com}
}
\and
Robert Krauthgamer%
  \thanks{Weizmann Institute of Science.
  Work partially supported by ONR Award N00014-18-1-2364, the Israel Science Foundation grant \#1086/18, and a Minerva Foundation grant.
    Email: \texttt{robert.krauthgamer@weizmann.ac.il}
  }
}

\maketitle

%\newpage
%\pagenumbering{arabic}
\begin{abstract}
We study approximation algorithms for variants of the \emph{median string} problem,
which asks for a string that minimizes
the sum of edit distances from a given set of $m$ strings of length $n$. 
Only the straightforward $2$-approximation is known for this NP-hard problem. 
This problem is motivated e.g.~by computational biology, 
and belongs to the class of median problems (over different metric spaces), 
which are fundamental tasks in data analysis. 

Our main result is for the Ulam metric,
where all strings are permutations over $[n]$
and each edit operation moves a symbol (deletion plus insertion). We devise for this problem an algorithms
that breaks the $2$-approximation barrier,
i.e., computes a $(2-\delta)$-approximate median permutation
for some constant $\delta>0$ in time $\tilde{O}(nm^2+n^3)$. We further use these techniques to achieve a $(2-\delta)$ approximation
for the median string problem in the special case
where the median is restricted to length $n$
and the optimal objective is large $\Omega(mn)$. 

We also design an approximation algorithm for the following
probabilistic model of the Ulam median: 
the input consists of $m$ perturbations of an (unknown) permutation $x$,
each generated by moving every symbol to a random position
with probability (a parameter) $\epsilon>0$.
Our algorithm computes with high probability
a $(1+o(1/\epsilon))$-approximate median permutation in time $O(mn^2+n^3)$.

%%% Local Variables:
%%% mode: latex
%%% TeX-master: "medianUlam"
%%% End:

\end{abstract}

\newpage

\section{Introduction}
\label{sec:intro}

One of the most common aggregation tasks in data analysis
is to find a representative for a given data set $S$,
often formulated as an optimization problem.
Perhaps the most popular version asks
to minimize the sum of distances from all the data points in $S$
(in a metric space relevant to the intended application). 
More formally,
the goal is to find $y$ in the metric space (not necessarily from $S$)
that minimizes the objective function 
$$
  \obj(S,y) := \sum_{x\in S} d(y,x), 
$$
and an optimal $y$ is called a \emph{median} (or a \emph{geometric median}).
For many applications, it suffices to find an \emph{approximate median},
i.e., a point in the metric space whose objective value
approximates the minimum (multiplicatively), 
see Section~\ref{sec:prelim} for a formal definition.
The problem of finding an (approximate) median has been studied extensively
both in theory and applied domains, over various metric spaces. The most well-studied version is over a Euclidean space (called the Fermat-Weber problem), for which currently the best algorithm finds a $(1+\epsilon)$-approximate median (for any $\epsilon > 0$) in near-linear time~\cite{cohen2016geometric} (see references therein for an overview).
Other metric spaces that have been considered for the median problem
include Hamming (folklore), the edit metric~\cite{Sankoff75, kruskal1983, NR03}, rankings~\cite{DKNS01,ACN08}, Jaccard distance~\cite{CKPV10},
and many more~\cite{fletcher2008robust, minsker2015geometric, cardot2017online}.

%The median problem for strings is extremely important in computational biology, where strings represent genomic or protein sequences, and the median represents a meaningful pattern. 
%especially in locating binding sites in unaligned sequences~\cite{SH89, hertz1995}, designing genetic probes~\cite{lanctot2003distinguishing}, universal PCR primer design~\cite{lucas1991improved, proutski1996primer, lanctot2003distinguishing}.
The median problem over the \emph{edit metric} 
(where the edit distance between two strings is the minimum number
of character insertion, deletion and substitution operations
required to transform one string to the other)
is called the \emph{median string} problem~\cite{kohonen1985median}
(an equivalent formulation is known as \emph{multiple sequence alignment}~\cite{gusfield1997}).
It finds numerous applications in many domains, 
including computational biology~\cite{gusfield1997, pevzner2000computational},
DNA storage system~\cite{GBCDLSB13, RMRAJY17},
speech recognition~\cite{kohonen1985median},
and classification~\cite{martinez2000use}. 

Given a set of $m$ strings each of length $n$,
the median string problem can be solved using standard
dynamic programming~\cite{Sankoff75, kruskal1983} in time $O(2^mn^m)$, 
and it is known to be NP-hard~\cite{HC00, NR03} (even W[1]-hard~\cite{NR03}). 
There is a folklore algorithm that easily computes a $2$-approximate (actually, $(2-\frac{1}{m+1})$-approximate) median --- 
simply report the best input string, i.e., $y^*\in S$ that minimizes the objective
(we call this algorithm \BI, see Procedure~\ref{alg:best-input})
--- and in fact this argument holds in every metric space. 
Although several heuristic algorithms exist~\cite{casacuberta1997greedy, kruzslicz1999improved, fischer2000string, pedreira2007spatial, abreu2014new, hayashida2016integer, Mirabal19}, no polynomial-time algorithm is known to break below $2$-approximation
(i.e., achieve factor $2-\delta$ for fixed $\delta > 0$)
for the median string problem. 
In contrast, over the Hamming metric a median can be computed in linear time
by simply taking a coordinate-wise plurality vote. 
One can also compute a $(1+\epsilon)$-approximation in sublinear time
using sampling (similarly to~\cite{FOR17}).

We focus mostly on approximating the median over the \emph{Ulam metric},
which is a close variant of the edit metric.
The Ulam metric of dimension $n$ is the metric space $(\sym_n,d)$,
where $\sym_n$ is the set of all permutations over $[n]$
and $d(x,y)$ is the minimum number of character moves needed
to transform $x$ into $y$~\cite{AD99}.%
\footnote{One may also consider one deletion and one insertion operation instead of a character move, and define the distance accordingly~\cite{CMS01}.
}
The importance of studying the Ulam metric is twofold.
First, it is an interesting measure of dissimilarity between rankings, 
which arise in application domains like sports, databases, and statistics.
Second, it captures some of the inherent difficulties of the edit metric,
and thus, any progress in the Ulam metric may provide insights to tackle
the more general edit metric.
The Ulam metric has thus been studied from different algorithmic
perspectives~\cite{CMS01, CK06, AK10, AN10, NSS17, BS19},
but unfortunately, no polynomial-time algorithm is currently known
to break below the folklore $2$-approximation (actually, $(2-\frac{1}{m+1})$-approximation for $m$ input permutations) bound for Ulam median. 
In contrast, for the median with respect to \emph{Kendall's tau distance} over permutations,
which is often used for rank aggregation~\cite{kemeny1959mathematics, young1988condorcet, young1978consistent, DKNS01}, a PTAS~\cite{kenyon2007rank, Schudy2012thesis} is known, improving upon a polynomial-time $4/3$-approximation~\cite{ACN08}.

Our main result is a deterministic polynomial-time algorithm 
that breaks below $2$-approximation for Ulam median
(see Section~\ref{sec:worst-polytime}).

\begin{restatable}{theorem}{worstpoly}
\label{thm:worst-approx-poly}
There is a constant $\delta>0$ and a deterministic algorithm that,
given as input a set of $m$ permutations $S\subseteq \sym_n$, 
computes a $(2-\delta)$-approximate median
in time $O(nm^2 \log n + n^2m +n^3)$. 
\end{restatable}

The running time's quadratic dependence on $m$ comes from
a naive subroutine to find the best median among the data set $S$.
We can replace this subroutine with a randomized $(1+\epsilon)$-approximation algorithm, due to~\cite{indyk1999sublinear}, to obtain linear dependence on $m$.

Furthermore, one of our key algorithmic ingredients for the Ulam metric
extends to the more general edit metric, albeit with some restrictions
on the length of the median string and on the optimal objective value.
Specifically, we refer to the following problem: 
Given a set of strings over $\Sigma^n$ (for an alphabet $\Sigma$),
find a string in $\Sigma^n$, called a \emph{length-$n$ edit-median},
that attains the minimum objective value under the edit metric. 
In fact, the improvement is achieved using the folklore algorithm mentioned earlier, 
i.e., in our restricted setting this algorithm actually beats $2$-approximation!
(See Section~\ref{sec:high-edit}.) 

\begin{restatable}{theorem}{worstpolyedit}
\label{thm:worst-approx-edit}
Given a set of strings $S \subseteq \Sigma^n$ 
whose optimal median objective value is at least $|S|n/c$ for some $c > 1$,
Procedure \BI reports a $(2-\frac{1}{50c^2})$-approximate length-$n$ edit-median
in time $O(nm^2 \log n)$.\footnote{We make no attempt to optimize the constants.}
\end{restatable}

Restrictions on the median string's length and on the optimum objective value
may be justified in certain applications.
For example, in DNA storage system~\cite{GBCDLSB13, RMRAJY17},
stored data is retrieved using next-generation sequencing,
and as a result several noisy copies of the stored data are generated.
(Note, here the noises are in the form of insertions, deletions and substitutions.)
%\rnote{Why is this relevant? this is standard edit distance} 
Currently, researchers use median-finding heuristics to recover the stored data from these noisy copies. Since third-generation sequencing technology like single molecule real time sequencing (SMRT)~\cite{RCS13} involves $12-18\%$ errors, the optimum median objective value is quite large (and matches our restriction). Moreover, since the noise is randomly added at each location during sequencing, it follows from standard concentration inequalities that with high probability the lengths of the noisy strings are "close" to that of the originally stored data (or the median). Thus, a length-restricted median,  as in our result, should be a good approximation of the original one. 

Motivated by the above application
we further investigate a probabilistic model for the Ulam metric, as follows.
The input consists of $m$ perturbations of an (unknown) permutation $x$,
each generated by moving every symbol to a random position
with probability (a parameter) $\epsilon>0$.
See Section~\ref{sec:avg-case} for a formal definition
of this input distribution, which we denote by $S(x,\epsilon,m)$.
We then provide a $(1+\delta)$-approximate median for this model.

\begin{restatable}{theorem}{avgpoly}
\label{thm:avg-case-poly}
Fix a parameter $\epsilon \in (0,1/40)$, a permutation $x\in \sym_n$,
and $40 \le m \le n$.
There is an $O(n^3)$-time deterministic algorithm that,
given input $S$ drawn from $S(x,\epsilon,m)$,
outputs a $(1+\delta)$-approximate median of $S$, for $\delta=\frac{20}{m}+\frac{3}{\log(n/\epsilon)} + \frac{2 e ^{-m/40}}{\epsilon}$, 
with probability at least $1-5/m$. 
% (The probability is over the input distribution.)
\end{restatable}
Our analysis is based on a novel encoding-decoding (information-theoretic) argument, which we hope could also be applied to the more general edit metric (left open for future work).

Even though the Ulam distance is a special case of the edit distance,
for the problem of finding an exact median over the Ulam metric,
no algorithm faster than exhaustive search (over the $n!$ permutations)
is known.
In contrast, for the edit metric, at least with constantly many input strings, 
one can find an exact median in polynomial time 
using dynamic programming~\cite{Sankoff75, kruskal1983}.
The lack of a polynomial-time algorithm for the Ulam metric,
even with constantly many inputs, is perhaps not very surprising,
because the related problem of rank aggregation,
which is the same median problem over permutations (rankings)
but with respect to Kendall's tau distance, 
is NP-hard even for $m=4$ permutations~\cite{DKNS01}.
And even for $m=3$ permutations, the current polynomial-time algorithm
achieves only $1.2$-approximation~\cite{ACN08}. 

Nevertheless (and in contrast to rank aggregation),
we show a polynomial-time algorithm that solves the Ulam median problem
for $m=3$ permutations (see Section~\ref{sec:worst-exptime}).
We further extend this result to show that for $m$ inputs there is an $O(2^{m+1} n^{m+1})$-time algorithm computing a $1.5$-approximate median.

\begin{remark}
Some literature slightly extend the notion of a permutation,
and call a string $x\in \Sigma^n$ a permutation
if it consists of distinct characters~\cite{CMS01, cormode2003sequence}.
Then the Ulam metric of dimension $n$ is defined over all these permutations,
and distances are according to the standard edit distance.
All our results hold also for this variant of the Ulam metric
as long as the goal is to find (as median) a permutation of length $n$.
However, for the sake of simplicity we present our results only for
the standard definitions of a permutation and the Ulam metric
(as in~\cite{AD99}).
\end{remark}

\subsection{Technical Overview}
\label{sec:overview} 

\paragraph*{Breaking below 2-approximation (in worst-case)}
We start with an overview of our main result, 
a $(2-\delta)$-approximation for the median under Ulam
(Theorem~\ref{thm:worst-approx-poly}). 
It is instructive to understand if and when does the well-known
$2$-approximation algorithm fail to achieve approximation better than $2$. 
Recall that this algorithm reports the best input permutation $y\in S$
(see \BI in Procedure~\ref{alg:best-input}). 
To analyze it, let $\med $ be an optimal median for $S$,
and let $y^*\in S$ be an input permutation that is closest to $\med$,
i.e., $d(y^*,\med) = \min_{x\in S} d(x,\med)$.
Then by the triangle inequality 
\begin{equation} \label{eq:2easy}
  \sum_{x\in S}d(y^*,x)
  \le \sum_{x\in S} [d(y^*,\med) +d(\med ,x)]
  \le 2 \sum_{x\in S} d(\med, x), 
\end{equation}
and the objective value of the reported $y\in S$ is only better. 

Now consider a scenario where this analysis is tight; 
suppose every input permutation is at the same distance $\ell>0$ from $\med$, 
and the distance between every two input permutations is $2\ell$.
Then the objective value for $\med$ is $\ell m$,
but for every input permutation $y\in S$ it is $2\ell m$.
In this scenario, a better approximation must exploit
the structure of the input permutations. 

Somewhat surprisingly, we show that if the optimal objective value
is large, say $\Omega(mn)$, then the above scenario cannot occur.
To gain intuition, start with a favorable case where all input permutations
are at distance at least $n/c$ from $\med $ (for some constant $c>1$).
Then in an optimal alignment of an input $x\in S$ with $\med$,
at least $\ell=n/c$ symbols are not aligned (i.e., are moved).
Now a combinatorial bound (based on the inclusion–exclusion principle),
implies that every $2c$ input permutations must include a pair $x',x''\in S$
whose sets of non-aligned symbols (with $\med$) have a large intersection,
specifically of size $\Omega(n/c^2)$.
This yields a non-trivial distance bound 
$d(x',x'') \leq 2n/c - \Omega(n/c^2) < 2\ell$
that contradicts our assumption. 
Our full argument (in Section~\ref{sec:high-regime}),
employs this idea more generally than just the favorable case  
(i.e., whenever the optimal objective value is large). 
We use additional steps, like averaging arguments
to exclude permutations that are too close or too far from $\med$,
and an iterative ``clustering'' of the input permutations
around at most $2c$ so-called candidate permutations, 
%and an extremal-combinatorics bound of Theorem~\ref{thm:intersect-family},
to infer that for at least one candidate $y^*\in S$,
the cluster around it is large.
We use this to bound the objective value for this candidate $y^*$,
but with a gain compared to~\eqref{eq:2easy},
due to the cluster around $y^*$ that has many permutations, 
all within distance $2n/c - \Omega(n/c^2)$ from $y$.
We conclude that reporting the best permutation among the input $S$
is at least as good as $y^*$ and thus breaks below $2$-approximation.
This argument about large optimal objective value
extends to the edit metric, i.e., over general strings
(see Section~\ref{sec:high-edit}).

The general case of the Ulam metric
(without assuming that the optimal objective value is large)
is more difficult and involved, but perhaps surprisingly,
reuses the main technical idea from above, although not in a black-box manner. 
We split this analysis into two cases by considering
the contribution of each symbol to the optimal objective value. 
To be more precise,
fix an optimal alignment of each input permutation $x\in S$ with $\med$,
and let the \emph{cost} of a symbol count in how many alignments 
(equivalently, for how many $x\in S$) this symbol is not aligned. 

Informally, one case (called Case 2 in Section~\ref{sec:general}) 
is when the cost is distributed over a few symbols.
Here, by restricting these optimal alignments to these costly symbols
we can employ a strategy similar to our first case above
(optimal objective value is large).
Intuitively, for the costly symbols we obtain approximation better than $2$,
and approximation factor $2$ for the other symbols, 
and altogether conclude that reporting the best permutation among the input $S$
breaks below $2$-approximation. 
While this plan is quite intuitive,
combining these two analyses into one argument is technically challenging
because we cannot really analyze these two sets of symbols
(denoted therein $G$ and $\overline G$) separately. 

In the remaining case (called Case 1 in Section~\ref{sec:general}), 
the cost is distributed over many symbols;
this is a completely different situation and we devise for it
an interesting new algorithm (\BRO in Procedure~\ref{alg:ptwo}).
The main idea is that now most of the symbols in $[n]$
must be aligned in many optimal alignments (i.e., for many $x\in S$),
and thus for every two such symbols, their relative order in $\med$
can be easily deduced from the input (by taking majority over all $x\in S$).
More precisely, call a symbol \emph{good}
if it is aligned in at least $0.9$-fraction of the input permutations.
Observe that every two good symbols must be aligned simultaneously
in at least $0.8$-fraction of the input permutations,
hence their relative order in $\med$ can be computed
by checking their order in each $x\in S$ and taking a majority vote.
This observation is very useful because in this case most symbols are good,
however the challenge is that we cannot identify the good symbols reliably.
Instead, our algorithm finds all pairs of symbols with a qualified majority 
(say, above $0.8$ threshold);
which is a superset of the aforementioned pairs,
and might contains spurious pairs (involving bad symbols)
that contradict the relative order between good symbols.
We overcome this by iteratively removing symbols that participate
in a contradiction: 
our algorithm builds a directed graph $H$,
whose vertices represent symbols and whose edges represent qualified majority,
and then iteratively removes a (shortest) cycle,
where removal of a cycle means removing all its vertices (not only edges).
We prove that every such cycle consists mostly of bad vertices/symbols,
and straightforward counting shows that
the final graph $\overline{H}$ contains almost all the good symbols.
Moreover, this final $\overline{H}$ contains no cycles,
and thus topological sort retrieves the order (according to $\med$) 
of almost all good symbols.
We then obtain a permutation that is pretty close to $\med$
by simply adding all the missing symbols at the end.
We point out that the approximation factor that we get here (Case 1)
can in principle be close to $1$ (it depends on some parameters). 
We indeed exploit this in our algorithm for the probabilistic model
(as discussed next),
however the balance with Case 2 is quite poor, 
and thus our overall approximation factor is quite close to $2$.

\paragraph*{Finding median in a probabilistic model. }
Our next result (in Section~\ref{sec:avg-case}) deals with a probabilistic model over the Ulam metric,
and is motivated by the application to DNA storage system. 
In this model, the input $S$ consists of $m$ permutations,
each generated from an unknown permutation $x$ by moving each symbol
independently with probability $\epsilon>0$ to a randomly chosen location.
Let $S(x, \epsilon, m)$ denote the distribution generated in this model.
Given an input $S$ drawn from this distribution,
the objective is to find a median of $S$ (not the unknown $x$).
We show a polynomial-time algorithm that finds (with high probability)
a $(1+o(1/\epsilon))$-approximate median of $S$ 
(see Theorem~\ref{thm:avg-case-poly} for the precise factor).
Our argument consists of two parts.
First, we show that the unknown $x$ is itself a $(1+ o(1))$-approximate median.
Second, we provide an algorithm that computes a permutation $\tilde{x}$
that is "very close" to the unknown $x$.
It then follows by the triangle inequality
that $\tilde{x}$ is an approximate median of $S$. 

The first part goes via an information-theoretic (encoding-decoding based) argument. We show that if $x$ is not an approximate median of $S$, then we can encode the set of (random) move operations used to generate $S$, using fewer number of bits than that required by the information-theoretic bound. It is evident from the generation process of each permutation $x_i \in S$,
that the (Shannon) entropy of this set of random move operations has total entropy about $\sum_{x_i \in S} d(x,x_i)$,
which is the median objective value for $x$.
Let $\med$ be an optimal median of $S$,
and denote the optimal median objective value
by $\opt(S)=\sum_{x_i \in S}d(x_i,\med)$.
To encode all the $x_i$'s (given $x$),
one can first specify a set of move operations to transform $x$ into $\med$, and then specify the move operations to transform $\med$ to each $x_i$.
The length of this encoding is about
$d(x,\med) + \sum_{x_i \in S}d(x_i, \med) = d(x,\med) + \opt(S)$.
If the above encoding could be used to recover all the random move operations,
then we could conclude, by Shannon's source coding theorem,
that the objective value with respect to $x$ is almost equal to $\opt(S)$, 
and thus $x$ is an approximate median.
We do not know if this encoding is indeed sufficient for the said decoding,
but we can add to it a little extra information,
that suffices to decode the set of random operations;
let us elaborate how works.

From the above encoding we know all the $x_i$'s. We show that almost none of the symbols (except about $O(\log n)$ many) that were moved from $x$ to generate $x_i$, appears in every \emph{longest common subsequence} (\lcs) between $x_i$ and $x$. Therefore by computing an {\lcs} between $x_i$ and $x$, all but $O(\log n)$ moved symbols can be identified. Note that a random move operation consists of two pieces of information, the moved symbol and the location where it is moved. From the {\lcs} we get back the information about the moved symbols. So the only task remains is to identify the locations where they are moved. Suppose a symbol $a$ is moved to a location right next to another symbol $b$ to generate $x_i$ from $x$. (Note, since we are dealing with permutations we can identify a location by its preceding symbol).
Observe that the symbol $b$ might also be moved, but only with probability $\epsilon$, and so with the remaining probability $b$ just precedes $a$ in $x_i$. Therefore for each of the moved symbols (except for about an $\epsilon$-fraction) for $x_i$, just by looking into its preceding symbol in $x_i$ we can identify its moved location. For the remaining $\epsilon$-faction we could encode the moved locations explicitly, but that would worsen the approximation factor. To handle this, we argue that for each of these $\epsilon$-fraction of moved symbols we can identify a $O(\log n)$-sized "set of candidate locations", and thus it suffices to encode the exact location only inside this candidate set using $O(\log \log n)$ bits. Now after a careful calculation we get that the whole encoding is of length $(1+o(1))\opt(S)$. Then we apply Shannon's source coding theorem and conclude that the objective value with respect to $x$ (which is equal to the total entropy of random move operations) is at most $(1+o(1))\opt(S)$, and so $x$ is a $(1+o(1))$-approximate median of $S$.

\iffalse
Note that from the above encoding, given $x$, one can easily recover all of $S$.
Next, we exploit the fact that each $x_i\in S$ is a permutation
and was generated from $x$ by moving each symbol
with probability $\epsilon$ to a random location.
By computing a longest common subsequence ({\lcs}) between $x$ and $x_i$
we can identify all but $O(\log n)$ of the moved symbols,
and we encode the remaining moved symbols using a few extra bits.
Furthermore, for all the, again except "a few", moved symbols by looking into their preceding symbols in $x_i$, we identify the random locations where they were moved, and thus almost the whole set $\Sigma_i^e$. For the remaining "exceptional" symbols, we identify a "tiny" ($O(\log n)$ sized) interval in $x$ where it was moved and then encode the exact moved location inside that tiny interval. Overall we use only a few extra bits to encode this extra information. So overall the total encoding length becomes $(1+o(1))\opt(S)$. Now we apply Shannon's source coding theorem and conclude that the total entropy of all the random move operations, which is roughly equal to the objective value with respect to $x$, is at most $(1+o(1))\opt(S)$.
\fi

The second step is to reconstruct the initial unknown permutation $x$.
The task is similar to that in~\cite{chierichetti2014},
although their underlying distance is Kendall's tau distance,
and their random perturbation model is different.
In our case, each symbol of $x$ is moved with probability $\epsilon$ to generate a permutation $x_i$.
Hence any particular symbol is moved in expectation in $\epsilon m$ many $x_i$'s.
Further, the total objective value is equally distributed on all the symbols. This scenario is similar to Case 1 in our worst-case approximation algorithm, 
except that now the underlying permutation is $x$ instead of $\med$.
Thus we can use Procedure {\BRO} discussed above,
and since the objective value is distributed equally among all the symbols,
we can find a permutation $\tilde{x}$ that is very close to the unknown $x$.
Moreover, when $m \ge \Omega(\log n)$
we show that for every two symbols $a \ne b \in [n]$
we can decide (with high probability) whether $a$ appears before $b$ in $x$ or not, by observing their relative order in the input permutations.
Hence using any sorting algorithm (with slight modification)
we can reconstruct $x$ with high probability.

\paragraph*{Exact median for three permutations. }
We present in Section~\ref{sec:worst-exptime} an algorithm that finds an exact median for three permutations.
The non-trivial part of this algorithm is that running
the conventional dynamic program for a median~\cite{Sankoff75, kruskal1983}
will compute a string, which need not be a permutation,
and in fact even its length need not be equal to $n$.
Therefore, we first use a slight modification of that dynamic program
to compute a string $x'$ of length exactly $n$
(but not necessarily a permutation)
with the minimum possible median objective value with respect to edit distance
(i.e., $x'=\argmin_{y\in [n]^n}\sum_{i\in[3]}\ed(y,x_i)$).
%This dynamic program is a modification of~\cite{Sankoff75, kruskal1983}.
Crucially, the objective value attained by this $x'$
is at most that of a median permutation $\med$.
Next, we post-process $x'$ to produce a permutation $\tilde{x}$ over $[n]$,
by removing multiple occurrences of any symbol
and then inserting all the missing symbols (in a careful manner).
To complete the analysis we show that
$\sum_{i\in[3]} \ed(\tilde{x},x_i)=\sum_{i\in [3]} \ed(x',x_i)$.

%%% Local Variables:
%%% mode: latex
%%% TeX-master: "medianUlam"
%%% End:

\subsection{Conclusion}
\label{sec:conclusion}

%We study the problem of finding a median over the Ulam (and edit) metric.
There is a folklore algorithm that computes $2$-approximate median in any metric space,
however no better approximation algorithm was known
for the Ulam and edit metrics, despite their utter importance.
Our main result breaks below $2$-approximation for the Ulam metric.
Further, we extend our result to the more general edit metric,
albeit with certain restrictions on the length of the median and on the optimal median objective value.
An exciting future direction,
is to beat $2$-approximation without these restrictions.
In fact, this was recently stated as an open problem~\cite{Bertinoro19}.

We also consider for the median Ulam problem a probabilistic model,
which is motivated by the applications to DNA storage system,
and we provide for it a $(1+o(1))$-approximation algorithm. 
In achieving our result, we use novel encoding-decoding (information-theoretic) argument, which we hope could also be used for the edit metric
(left open for the future work) and perhaps even more general metric spaces.

%%% Local Variables:
%%% mode: latex
%%% TeX-master: "medianUlam"
%%% End:

\section{Preliminaries}
\label{sec:prelim}
\textbf{Notations:} Let $[n]$ denote the set $\{1,2,\cdots,n\}$. We refer the set of all permutations over $[n]$ by $\sym_n$. Throughout this paper we consider any permutation $x$ as a sequence of numbers $a_1,a_2,\cdots,a_n$ such that $x(i)=a_i$. For any subset $I\subseteq [n]$, let $x(I) := \{x(i)|i \in I\}$.

\paragraph*{The Ulam Metric and the Problem of Finding Median. }Given two permutations $x,y \in \sym_n$, the \emph{Ulam distance} between them, denoted by $d(x,y)$, is the minimum number of character move operations that is needed to transform $x$ into $y$.

Given two strings (permutations) $x$ and $y$ of lengths $n_1$ and $n_2$ respectively, \emph{alignment} $g$ is a function from $[n_1]$ to $[n_2]\cup \{*\}$ which satisfies:
\begin{itemize}
\item $\forall i\in[n_1],\text{if } g(i)\neq *, \text{ then } x(i)=y(g(i)) $;
\item For any two $i \ne j \in [n_1]$, such that $g(i)\neq *$ and $ g(j)\neq *$, if $i>j$, then $g(i)>g(j)$.
\end{itemize}

For an alignment $g$ between two strings (permutations) $x$ and $y$, we say $g$ \emph{aligns} a character $x(i)$ with some character $y(j)$ iff $j=g(i)$.

Given a set $S \subseteq \sym_n$ and another permutation $y \in \sym_n$, we refer the quantity $\sum_{x \in S}d(y,x)$ by the \emph{median objective value} of $S$ with respect to $y$, denoted by $\obj(S,y)$. 

Given a set $S \subseteq \sym_n$, a \emph{median} of $S$ is a permutation $\med \in \sym_n$ (not necessarily from $S$) such that $\obj(S,\med)$ is minimized, i.e., $\med = \argmin_{y \in \sym_n} \obj(S,y)$. We refer $\obj(S,\med)$ by $\opt(S)$. We call a permutation $\tilde{x}$ a \emph{$c$-approximate median}, for some $c>0$, of $S$ iff $\obj(S,\tilde{x})\le \opt(S) \le c \cdot \obj(S,\tilde{x})$.

\paragraph*{A Folklore 2-approximation Algorithm. }For the problem of finding median (over any metric space, and so for the Ulam), there is a folklore 2-approximation algorithm (actually, a $(2-\frac{1}{m+1})$-approximation algorithm for $m$ input permutations). We briefly present here this algorithm for a set of permutations. We also refer to this algorithm as Procedure {\BI} (Procedure~\ref{alg:best-input}).

\begin{algorithm}
	\begin{algorithmic}[1]
		\REQUIRE $S\subseteq \mathcal{S}_n$.
		
		\ENSURE A permutation $y \in S$.
		
		% \vspace{1mm}
		% \hrule\vspace{1mm}
		
		\STATE For all pairs of permutations $x_i,x_j \in S$, compute $d(x_i,x_j)$.
		
		\RETURN $\argmin_{y\in S} \sum_{x\in S} d(y,x)$.

		\caption{{\BI} $(S)$}
		\label{alg:best-input}
	\end{algorithmic}
\end{algorithm}

\section{Breaking below 2-approximation (in Worst-case)}
\label{sec:worst-polytime}

In this section, we describe a polynomial-time algorithm that computes,
for any given input permutations, 
a $(2-\delta)$-approximate median under the Ulam metric.

\worstpoly*

We start with the description of our algorithm and the running time bound. Next we analyze the approximation factor into two parts. First we consider a special case where the objective is large and give a stronger approximation guarantee. After that we discuss the general case.
Given as input a set of permutations $S\subset \sym_n$,
our algorithm %{\AM} (Algorithm~\ref{alg:main-approx}) 
runs two procedures, each producing a permutation (candidate median),
and returns the better of the two (that has smaller objective value). 
%The main algorithm outputs the string with the minimum median objective value among the two strings returned by Procedure {\BI} and Procedure {\BRO}.
The first procedure is {\BI} (see Procedure~\ref{alg:best-input}), which reports an input permutation $y \in S$
that has the minimum objective value among all input permutations, 
i.e., $\argmin_{y\in S} \sum_{x\in S} d(y,x)$.
(We have discussed in previous section that this algorithm is well-known to achieve $2$-approximation
in every metric space.) 

The second procedure, called {\BRO}, %{\BRO} 
is given a parameter $0\le \alpha\le 1/10$, %(the value of which will be set later),
and works as follows (see also Procedure~\ref{alg:ptwo}). 
First, create a directed graph $H$ with vertex set $V(H)=[n]$ and edge set 
\[
  E(H)=\{(i,j): \text{$i$ appears before $j$ in at least $(1-2\alpha)|S|$ permutations in $S$} \}.
\]

Next, as long as the current graph $H$ is not acyclic,
repeatedly find in it a cycle of minimum length 
and delete all its vertices (with all their incident edges). 
Denote the resulting acyclic graph by $\overline{H}$,
and use topological sort to compute an ordering $\mathcal{P}$ 
of its vertex set $V(\overline{H})\subset V(H) = [n]$. We shall write $i\triangleleft j$ to denote that $i$ precedes $j$ in this ordering $\mathcal{P}$.
%Notice this partial ordering can be found trivially in time $O(|V_{\overline{H}}|^4)=O(n^4)$ We further extend this partial order to a complete order $\mathcal{C}$ as follows: For each $i,j\in\tilde{S}$ if neither $i\triangleleft j$ nor $j\triangleleft i$, add an arbitrary order $i\triangleleft j$. 
Let the string $\bar x$ be a permutation of
(set of symbols) $V(\overline{H})$ by ordering them according to $\mathcal{P}$.
Finally, output the permutation $\tilde{x}$ of $[n]$
that is obtained by appending to $\bar x$
all the remaining symbols $[n]\setminus V(\overline{H})$ in an arbitrary order. 

%Next we enumerate over all pairs of vertices in lexicographically increasing order. While enumerating, find a pair of vertices $v_i,v_j$ such that there exists a cycle involving both $v_i$ and $v_j$ in $H$. Let $p_{ij}$ be a shortest path from $v_i$ to $v_j$. Delete all the vertices appearing on path $p_{ij}$ (including both $v_i$ and $v_j$). Again find a pair of vertices such that there exists a cycle involving both of them in the modified graph. Continue the above deletion process until the resultant graph $\tilde{G}$ becomes acyclic.
 %Next find a topological sorted ordering of the vertices of $V_{\tilde{G}}$(note $\tilde{G}$ is a DAG). Let $\overline{x}$ be the string over the symbols associated with $V_{\tilde{G}}$ obeying this ordering. Create a permutation $\tilde{x}$ by appending the symbols (associated with the vertices) of $V\setminus V_{\tilde{G}}$ at the end of string $\overline{x}$ in any arbitrary order. Output the permutation $\tilde{x}$.

\begin{algorithm}
	\begin{algorithmic}[1]
		\REQUIRE $S\subseteq \mathcal{S}_n$ of size $m$, $0<\alpha \le 1/10$.
		
		\ENSURE A permutation string $\tilde{x}$ over $[n]$.
		
		% \vspace{1mm}
		% \hrule\vspace{1mm}
		
		\STATE $H\gets ([n],E)$ where
                $E= \{(i,j): \text{$i$ appears before $j$ in $\geq (1-2\alpha)|S|$ permutations in $S$} \}$  
	
		%\STATE Sort $P=\{(v_i,v_j)\in V\times V\}$ in lecxichographical increasing order. 

		\WHILE{\text{$H$ contains a cycle}}
		
		%\FOR {$i=1,\dots,n$}
	
		%\STATE $\mathcal{C}_i\leftarrow \min$ length cycle containing $i$ in $H$.
		
		%\STATE $\mathcal{C}=\mathcal{C}\cup \mathcal{C}_i$.
	
		%\ENDFOR
		
		%\STATE $\mathcal{C}_{\min}=\argmin_{\mathcal{C}_i \in\mathcal{C}} |\mathcal{C}_i|$.
		
		\STATE $\mathcal{C}_{\min}\gets $ {cycle of minimum length in $H$}  
		
		\STATE $H=H - V(\mathcal{C}_{\min})$ 
		
		\ENDWHILE
		
		%\STATE Defind order $\mathcal{P}$ on $\tilde{S}$: $\forall (i,j)\in \tilde{S}\times \tilde{S}$, $i$ preceeds $j$ if $v_i$ preceeds $v_j$ in $\mathcal{O}$.
		
                \STATE $\overline{H} \gets H$
		
		\STATE $\overline{x} \gets $ string formed by topological ordering of $\overline{H}$
		
		\STATE $\tilde{x} \gets $ string formed by appending to  $\overline{x}$ the symbols $[n]\setminus V(\overline{H})$ in an arbitrary order.

		\RETURN $\tilde{x}$.

		\caption{{\BRO} $(S,\alpha)$}
		\label{alg:ptwo}
	\end{algorithmic}
\end{algorithm}

%\rnote{I polished the pseudocode (do not capitalize, and no period).}
%\rnote{It would more reasonable to swap the notations of $\tilde x$ and $\overline x$ (to keep analogy to $\tilde H$). }

\paragraph*{Running time analysis. }
Let $m=|S|$.
Since $d(x,y)$ can be computed in $O(n \log n)$ time for any $x,y \in \sym_n$,
Procedure {\BI} runs in time $O(nm^2 \log n)$. 

In Procedure {\BRO}, given the set $S$ and parameter $\alpha$ we can compute graph $H$ in time $O(n^2m)$.
Next, we iteratively find a minimum-length cycle $\mathcal{C}_{\min}$
in the current graph $\tilde H$ in time $O(n^3)$ (using an All-Pairs Shortest Path algorithm),
and delete all the vertices of $\mathcal{C}_{\min}$ (and edges incident on these vertices) in time $O(n^2)$.
%we compute set $\mathcal{C}$ by computing for each 
%vertex, the shortest cycle containing it in the current modified graph. This can be done using DFS traversal in time $O(n^2)$ per vertex. Hence for all $n$ vertices, the time taken is $O(n^3)$. After computing set $\mathcal{C}$, we find the minimum length cycle in it and delete all vertices in it. As $|\mathcal{C}|\le n$ and length of each cycle is $O(n)$, this can be done in time $O(n^2)$.
Hence, each iteration takes time $O(n^3)$,
and since the number of iterations required is at most $n$,
the total time to compute $\overline{H}$ is $O(n^4)$. 
Now since $\overline{H}$ has at most $n$ vertices,
computing a topological order of its vertices runs in time $O(n^2)$.
Given this ordering, the strings $\overline{x}$ and $\tilde{x}$
are computed in time $O(n)$. 
Thus, Procedure {\BRO} runs in time $O(n^2m+n^4)$. 

As our main algorithm outputs the string with the minimum median objective value among the two strings returned by Procedure {\BI} and Procedure {\BRO},
its total running time is $O(nm^2\log n+n^2m+n^4)$. In Remark~\ref{rem:time-improve}, we will comment on how to improve the $O(n^4)$ factor of the above time-bound to $O(n^3)$ by slightly modifying Procedure {\BRO}.

We devote the remaining part of the section to derive the approximation ratio of our algorithm. We will first consider a special case when $\opt(S)$ is "large", for which the analysis is slightly simpler, and also, we get a stronger approximation guarantee. Then we will turn our attention to the more general case. Although the result for the high regime is independent of that for the general case, one of the main ideas carries forward to the general case, albeit with more complications.

\subsection{High regime of the optimal objective value}
\label{sec:high-regime}

\begin{lemma}
\label{lem:high-regime-ulam}
Given a set of permutations $S \subseteq \sym_n$
with $\opt(S) \ge |S|n/c$ for some $c > 1$, Procedure {\BI($S$)} outputs a $(2-\frac{1}{50c^2})$-approximate median. 
\end{lemma}

\begin{proof}
We first introduce some notation.
Let $m=|S|$ and set $\delta=\frac{1}{50c^2}$.
Let $\med$ be an (optimal) median of $S$; 
then $\opt(S) = \sum_{x \in S}d(x,\med)$,
and for brevity we denote it by $\opt$. 
For any subset $S'\subseteq S$,
denote $\opt_{S'} = \sum_{x \in S'}d(x,\med)$. 
We assume henceforth that
\begin{equation} \label{eq:noneclose}
  \forall x\in S, \qquad
  d(x,\med) > (1-\delta)\opt/m ,
\end{equation}
because any $x'\in S$ that violates~\eqref{eq:noneclose}
is a $(2-\delta)$-approximate median of $S$ by the triangle inequality,
formally 
$
  \sum_{z\in S} d(x',z)
  \le \sum_{z\in S} [d(x',\med) + d(\med,z)]
%  \le m\cdot (1-\delta)\opt/m + \opt
  \le (2-\delta) \opt
$. 

For each $x \in S$ fix an optimal alignment (see Section~\ref{sec:prelim} for the definition) between $\med$ and $x$,
and denote by $I_x\subset[n]$ the set of symbols moved (i.e., not aligned) by this alignment. 
Then by~\eqref{eq:noneclose} we have
$|I_x| = d(x,\med) > (1-\delta)\opt/m \ge (1-\delta)n/c$.
Set $c'=\lceil \frac{c}{1-\delta}\rceil$
and $\xi=\frac{1}{2c'^2} $.

We now partition $S$ into the far and close permutations (from $\med$).
Let $F=\{ x\in S: d(x,\med) \ge (1+\delta)\opt/m \}$ 
and $\overline{F}=S\setminus F$.
Since $\opt = \sum_{x \in S}d(x,\med)$, by our assumption~\eqref{eq:noneclose}, $|F|\le |\overline{F}|$. Thus $|\overline{F}| \ge m/2$.
It follows that 
\begin{equation}
\label{eq:high-mass-close-large}
  \opt_{\overline{F}}
  \ge \tfrac{m}{2} \cdot (1-\delta) \tfrac{\opt}{m}
  = \tfrac{1-\delta}{2} \opt.
\end{equation}

Next, we partition $\overline{F}$ even further using the following procedure. 
Initialize a set $C=\emptyset$, 
and then iterate over the permutations $x\in \overline{F}$
in non-decreasing order of $|I_x|$. 
For each such $x$, if
\[
  \forall y \in C, \quad
  |I_x \cap I_y| < \xi n,
\]
then add $x$ to $C$ and create its ``buddies set'' $B_x=\emptyset$;
otherwise, pick $y \in C$ that violates the above, breaking ties arbitrarily, 
and add $x$ to its buddies set $B_y$. 
Note that this partitioning is solely for the sake of analysis.
Since $\overline{F}$ is processed in sorted order, it is clear that
\begin{equation}
\label{eq:small-center}
  \forall y\in C, x\in B_y, \qquad
  |I_y| \le |I_x|.
\end{equation}

We shall now prove two claims about this partitioning;
the first one argues that at least a buddies set $B_y$ (i.e., one ``cluster'') 
must be responsible for a large portion of the cost,
and the second one bounds the distances from its ``center'' $y$. 
%We will then use these claims to prove Lemma~\ref{lem:high-regime-ulam}.
%\rnote{I added this quick text to explain that the proof relies on two claims.}

\begin{claim} \label{clm:large-cluster-mass} 
There exists $y \in C$ such that 
\begin{equation}
\label{eq:large-cluster-mass}
%  \exists y \in C, \quad
  \opt_{B_y}
  \ge \frac{\opt_{\overline{F}}}{|C|} 
  \ge \frac{\opt_{\overline{F}}}{2c'}.
\end{equation}
\end{claim}

To prove the claim, we shall need the following upper bound
on the size of a family of subsets with small pairwise intersections. 
\begin{lemma}
\label{lem:intersect-family}
For every $n,c' \in \mathbb{N}$ and $0 < \xi \le \frac{2}{2c'^2}$, 
every family of subsets of $[n]$
in which every subset has size $n/c'$
and every pair of subsets share at most $\xi n$ elements,
has size at most $2c'$. 
\end{lemma}
We defer the proof of the above lemma to the end of this subsection. Now assuming the lemma we prove Claim~\ref{clm:large-cluster-mass}.

\begin{proof} [Proof of Claim~\ref{clm:large-cluster-mass}] 
Lemma~\ref{lem:intersect-family} applies to the set $C$ % Lemma~\ref{lem:intersect-family},
because $\xi=\frac{1}{2c'^2}$,
and by construction of $C$, 
all distinct $y,y'\in C$ satisfy $|I_y \cap I_{y'}| < \xi n$.
We thus conclude that 
\begin{equation*}
%\label{eq:cluster-points}
  |C| \le 2c'.
\end{equation*}
Now since $\overline{F}=\bigcup_{x \in C}B_x$,
a straightforward averaging implies the claim. 
\end{proof}

\begin{claim} \label{clm:distance-cluster}
Suppose $y\in C$ satisfies~\eqref{eq:large-cluster-mass}. 
Then its distance to every $x\in S$ is bounded by:
\begin{align}
  \forall x \in F, \quad 
  d(x,y) &\le 2 d(x,\med).
    \label{eq:distance-noncluster1}
  \\
  \forall x \in \overline{F}, \quad
  d(x,y) &\le (2+4\delta)d(x,\med)
    \label{eq:distance-noncluster2}
  \\
  \forall x \in B_y, \quad
  d(x,y) &\le (2-\rho)d(x,\med)
         & \text{where $\rho=\tfrac{(1-\delta)(c'-1)}{2(1+\delta)c'^2}$} .
   \label{eq:distance-cluster}
\end{align}
\end{claim}

\begin{proof} 
To prove~\eqref{eq:distance-noncluster1}, consider $x \in F$. 
Since $y\in C \subseteq \overline{F}$ we have $d(y,\med) \le d(x,\med)$,
and thus by the triangle inequality,
$d(x,y)\le d(x,\med)+d(y,\med) \le 2 d(x,\med)$. 

To prove~\eqref{eq:distance-noncluster2}, consider $x \in \overline{F}$.
Since $y\in C \subseteq \overline{F}$
and using our assumption~\eqref{eq:noneclose},
we have
$d(y,\med) \le (1+\delta)\opt/m \le \frac{1+\delta}{1-\delta}d(x,\med)$.
Using $\delta \le 1/2$ and the triangle inequality, we obtain
$d(x,y)\le d(x,\med)+d(y,\med) \le 2 (1+2\delta) d(x,\med)$. 

To prove~\eqref{eq:distance-cluster}, consider $x \in B_y$. Then 
\begin{align*}
  d(x,y)
  &\le |I_x|+|I_y|-|I_x \cap I_y|
  \\
  &\le 2 |I_x| - \xi n
  &\text{by~\eqref{eq:small-center}}
  \\ 
  &\le \Big(2-\frac{\xi n}{d(x,\med)}\Big)d(x,\med)
  &\text{since $|I_x| = d(x,\med)$}
  \\
  &\le \Big(2-\frac{(1-\delta)(c'-1)}{2(1+\delta)c'^2}\Big)d(x,\med)
\end{align*}
where the last inequality follows because $d(x,\med) \le (1+\delta)\opt/m \le (1+\delta)n/c$ and $c'=\lceil \frac{c}{1-\delta}\rceil$.
\end{proof}

We can now complete the proof of the lemma. %Lemma~\ref{lem:high-regime-ulam}.
Let $y\in C$ be as in Claims~\ref{clm:large-cluster-mass}
and~\ref{clm:distance-cluster} 
\begin{align*}
  \sum_{x \in S}d(x,y)
  &\le \sum_{x \in F} d(x,y) + \sum_{x \in \overline{F} \setminus B_y}d(x,y) + \sum_{x \in B_y}d(x,y)
    \\
  &\le 2 \opt_{F} + (2+4\delta)\opt_{\overline{F}\setminus B_y} + (2-\rho)\opt_{B_y}
  &\text{by Claim~\ref{clm:distance-cluster}}
  \\
  &\le 2 \opt + 4\delta \opt_{\overline{F}} - \rho \opt_{B_y}
  \\
  &\le 2 \opt + 4\delta \opt_{\overline{F}} - \rho \frac{\opt_{\overline{F}}}{c'}
  &\text{by~\eqref{eq:large-cluster-mass}}
  \\
  &\le 2\opt - (\frac{\rho}{c'}-4\delta) (1-\delta) \frac{\opt}{2}
  &\text{by~\eqref{eq:high-mass-close-large}}
  \\
  &\le \Big(2-\frac{(\frac{\rho}{c'}-4\delta)(1-\delta)}{2}\Big)\opt
  \\
  &\le (2-\frac{1}{50c^2})\opt
  &\text{for }\delta=\frac{1}{50c^2} .
\end{align*}
%\rnote{We can slightly simplify the third line above (without losing much) by ``omitting'' the second $4\delta$ term}
%\rnote{We can slightly improve the fourth line by keeping $4\delta \opt_F$(rather than bounding it by $4\delta \opt_{\bar F}$) }
This concludes the proof of Lemma~\ref{lem:high-regime-ulam}.
% In the last inequality we use the fact that $\frac{(1-\delta)^5}{1+\delta} \ge 1/2$ for our setting of $\delta$ for any $c>1$.
\end{proof}

It only remains to prove Lemma~\ref{lem:intersect-family}.
\begin{proof}[Proof of Lemma~\ref{lem:intersect-family}]
For contradiction sake, assume that there are $2c'$ subsets $Z_1,\cdots,Z_{2c'} \subseteq [n]$, such that 
\begin{align}
\label{eq:each-subset}
\forall i\in [2c'],\;& |Z_i| \ge n/c',\\
\label{eq:pair-subset}
\forall i\ne j \in [2c'], & |Z_i \cap Z_j| \le n/2c'^2.
\end{align}
Clearly, $\Big|\bigcup_{i \in [2c']} Z_i\Big| \le n$. Now from simple inclusion-exclusion principle together with~\eqref{eq:each-subset} and~\eqref{eq:pair-subset}, we get
\begin{align*}
\Big|\bigcup_{i \in [2c']} Z_i\Big| & \ge \frac{n}{c'}2c' - \frac{n}{2c'^2}{2c' \choose 2}\\
& = n + \frac{n}{2c'} > n
\end{align*}
which leads to a contradiction. Now the lemma follows.
\end{proof}

\subsection{The general case}
\label{sec:general}

We now argue the below 2-approximation guarantee for general input. 
\label{sec:small-regime}
Let us first recall a few notations from the last subsection and introduce a few more. Let $\med$ be an (arbitrary) median of $S$; then $\opt(S)=\sum_{x \in S}d(x,\med)$, and for brevity we denote it by $\opt$. For any subset $S'\subseteq S$ let $\opt_{S'} = \sum_{x \in S'}d(x,\med)$. Let us take parameters $\delta, \alpha, \beta, \gamma, \xi, \eta$, the value of which will be set later. (Note, the parameters $\delta, \xi$ were also used in the last subsection, but their values will be set differently in this subsection.)

From now on we assume that
\begin{equation} \label{eq:noneclose-low}
  \forall x\in S, \qquad
  d(x,\med) > (1-\delta)\opt/m ,
\end{equation}
because any $x'\in S$ that violates~\eqref{eq:noneclose-low}
is a $(2-\delta)$-approximate median of $S$ (by the triangle inequality). 

 For each $x \in S$ consider an (arbitrary) optimal alignment $g_x$ between $\med$ and $x$, and let $I_x$ denote the set of symbols that are moved (i.e., not aligned) by this alignment. Note, $|I_x|=d(x,\med)$. For any $x \in S$ and subset of symbols $Z\subseteq [n]$, let $I_x(Z)= I_x \cap Z$. 
 
 For each symbol $a \in [n]$ and any subset $S'\subseteq S$, let 
 $$c_{S'}(a)=|\{x \in S' : a \text{ is moved by the alignment }g_x\}|.$$  For brevity when $S'=S$ we drop the subscript $S'$ and simply use $c(a)$. For any subset $Z \subseteq [n]$ and $S'\subseteq S$, let $\opt_{S'}(Z)=\sum_{a \in Z}c_{S'}(a)$. Again for brevity when $Z=[n]$ we only use $\opt_{S'}$.

 We call a symbol $a \in [n]$ \emph{good} if $c(a) \le \alpha m$; otherwise \emph{bad}. Let 
 $$G=\{a \in [n] : a \text{ is a good symbol}\},$$ 
 and $\overline{G}=[n]\setminus G$. Now we divide our analysis into two cases depending on the size of $\overline{G}$.

\subsubsection*{Case 1: $|\overline{G}| \le \beta \frac{\opt}{m}$}

\begin{lemma}
\label{lem:low-regime-ordering}
Let $\alpha \in (0,1/10]$ and $\beta \in (0,1)$. Given a set $S \subseteq \sym_n$ of size $m$ such that the set of bad symbols $\overline{G}$ is of size at most $\beta \frac{\opt}{m}$, Procedure {\BRO}($S,\alpha$) %{\BRO($S,\alpha$)} 
outputs a $(1+\beta(1+8\alpha))$-approximate median.
\end{lemma}

In this section we show Procedure {\BRO}($S,\alpha$) finds a string $\tilde{x}$ such that $d(\tilde{x},\med)$

\noindent$\le \frac{1}{1-4\alpha}|\overline{G}|$. Given set $S$, Procedure {\BRO}() starts with the construction of the alignment graph $H=(V(H),E(H))$. Call a vertex \emph{good} if its corresponding symbol is good; otherwise call it \emph{bad}. We first make the following observation.

\begin{observation}
	\label{diedge}
	Given a set $S\subseteq \mathcal{S}_n$ of size $m$ and a parameter $0< \alpha\le 1/10$, let $G$ be the set of good symbols. 
	For each pair of symbols $i,j\in G$ there exists either a directed edge $(i,j)$ or $(j,i)$ in $E(H)$.
\end{observation}
\begin{proof}
	As both $i$ and $j$ are good symbols, $c(i)$ and $c(j)$ are at most $\alpha m$. Now for the sake of contradiction, assume the observation is not correct. Then neither $i$ precedes $j$, nor $j$ precedes $i$ in at least $(1-2\alpha)m$ strings of $S$. In this case, irrespective of the order of $i$ and $j$ in $\med$, together they can be aligned in less than $(1-2\alpha)m$ strings of $S$. Hence $c(i)+c(j)>2\alpha m$. But then at least one of $c(i)$ and $c(j)$ is strictly larger than $\alpha m$ and we get a contradiction.	
\end{proof}

Next we take the graph $H$ and repetitively delete the shortest length cycle until the resultant graph $\overline{H}=(V(\overline{H}),E(\overline{H}))$ becomes acyclic. We make the following claim.

\begin{claim}
	\label{order}
	Given a set $S\subseteq \mathcal{S}_n$ of size $m$ and a parameter $0 < \alpha\le 1/10$, let $H$ be the associated alignment graph. %Enumerate the set $P=\{(v_i,v_j);(v_i,v_j)\in V_H\times V_H\}$ in lexichographically increasing oredr. 
	Let ${H^k}$ be the graph obtained from ${H}$ after $k$ deletion steps and $\mathcal{C}^k_{min}$ be a shortest length cycle in ${H^k}$. 
	%of pairs that are not scanned after $k$ deletion steps. 	Let $(v_i^k,v_j^k)\in P^k$ be such that, there exists a directed path from $v_i^k$ to $v_j^k$ in ${H^k}$ and $(v_j^k,v_i^k)\in E_{H^k}$. 
	%Let $p_{ij}^k$ be the shortest path from $v_i^k$ to $v_j^k$. 
	Then for any $k\ge 0$, we claim the following.
	\begin{enumerate}
		
		\item Each cycle $\mathcal{C}^k$ of  ${H^k}$ has length at least $\frac{1}{2\alpha}$.
		
		\item There exist at most two good vertices in $\mathcal{C}^k_{min}$. 
	\end{enumerate}
\end{claim}

\begin{proof}
	Consider a cycle $\mathcal{C}^k$ in $H^k$. Let $i$ be some vertex and $(j,i)$ be some edge contained in $\mathcal{C}^k$. Without loss of generality assume $\mathcal{C}^k$ be the shortest cycle containing $i$. %Let $j$ be the vertex appearing just before $v_i^k$ in $\mathcal{C}^k$ and 
	Let $p_{ij}$ be the path from $i$ to $j$ in $\mathcal{C}^k$. Note, $p_{ij}$ is indeed the shortest path from $i$ to $j$. To prove the first part we show if length of $p_{ij}$  is $\ell$ then in at most $2\ell \alpha m$ strings of $S$, $j$ precedes $i$. We prove this by induction on the length of $p_{ij}$. As a base case, consider the scenario when the path length is just one, that is there is a directed edge from $i$ to $j$. Hence, in at least $(1-2 \alpha)m$ strings of $S$, $i$ precedes $j$ and therefore in at most $2 \alpha m$ strings $j$ precedes $i$. Let the claim be true for path of length $\ell -1$. Now consider a shortest path $i=i_1 \rightarrow i_2 \rightarrow \dots \rightarrow i_{\ell}\rightarrow i_{\ell+1}=j$ of length $\ell$. Notice the length of the shortest path between  $i$ and $i_{\ell}$ is $\ell -1$. Hence in at most $2(\ell-1) \alpha m$ strings, $i_{\ell}$ precedes $i$. Now as there is a directed edge from $i_{\ell}$ to $j$, in at least $(1-2\alpha)m$ strings $i_{\ell}$ precedes $j$. Together we claim in at most $2(\ell-1) \alpha m+2\alpha m=2\ell \alpha m$ strings, $j$ precedes $i$. Now as there is a directed edge from ${j}$ to ${i}$, $2\ell \alpha m\ge (1-2\alpha)m$. So, $\ell \ge \frac{1-2\alpha}{2\alpha}$. Hence length of the cycle is at least $\ell+1 \ge \frac{1}{2\alpha}$.
	
	%To prove the second part notice in the first claim we show length of $p_{ij}^k$ is at least $\frac{1-2\alpha}{2\alpha}$. Hence any cycle will have length at least $\frac{1-2\alpha}{2\alpha}+1$. As $\alpha\le 1/10$, any cycle will have length at least $5$.
	
	Fix two consecutive vertices $i,j$ in $\mathcal{C}_{\min}^k$ such that the directed edge $(j,i)$ is part of $\mathcal{C}_{\min}^k$.
	To prove the second part assume there are more than two good vertices, namely $v_1, v_2,\dots v_{\ell'}$ appearing on $p_{ij}$. Moreover, assume they appear on the path $p_{ij}$ in the above order. By Observation~\ref{diedge} between any $v_q$ and $v_r$ there is an edge. 
	First we claim 
	for each pair $q,r\in[\ell']$ where $q<r$, except both $v_q=i$ and $v_r=j$, the direction of the edge is from $v_q$ to $v_r$. As otherwise let $\exists q<r$ where either $v_q\neq i$ or $v_r\neq j$ or both, the edge is from $v_r$ to $v_q$. This gives rise to a cycle $(v_q,\dots,v_r,v_q)$ which has length strictly smaller than the length of $\mathcal{C}_{\min}^k$, and thus we get a contradiction. Next we divide the proof into two cases.
	
	\textbf{Case i: (When at least one of $i$ and $j$ is a bad symbol)} We have already seen, $\forall q,r\in[\ell']$ where $q<r$ the edge is from $v_q$ to $v_r$. Following this there exists a directed edge from $v_1$ to $v_{\ell'}$. Hence, the concatenation of the path from ${i}$ to $v_1$, the edge $(v_1,v_\ell')$ and the path from $v_{\ell'}$ to ${j}$ creates a path from $i$ to $j$ of length $|p_{ij}|-(\ell'-2)<|p_{ij}|$ as $\ell'>2$, and we get a contradiction as we assumed $\mathcal{C}_{\min}^k$ to be the shortest length cycle. 
	
	\textbf{Case ii: (When both $i$ and $j$ are  good symbols)} In this case, as we have already argued there must be an edge from $v_1=i$ to $v_2$ and an edge from $v_2$ to $v_{\ell'}=j$. That implies there is a cycle $(i,v_2,j,i)$ of length $3$, which contradicts the first part of our claim that says each cycle must be of length at least $\frac{1}{2\alpha}\ge 5$ (for $\alpha\le 1/10$).
	%%%%%%%%Following the claim there exists a directed edge from $v_1^k$ to $v^*$ where $v^*=v_\ell^k$ if $\ell\neq j$ otherwise $v^*=v_{\ell-1}^k$. But then the concatenation of the path from $v_{i}^k$ to $v_1$, the edge $(v_1,v^*)$ and the path from $v^*$ to $v_{j}^k$ creates a path from $v_{i}^k$ to $v_{j}^k$ of length $|p_{ij}^k|-(\ell-3)<|p_{ij}^k|$ if $\ell>3$ and we get a contradiction as we assumed $\mathcal{C}_{\min}^k$ to be the shortest length cycle. Otherwise if $\ell=3$ we get a cycle of length $3$ which contradicts claim\ref{order} part one.
	% Otherwise let $v_\ell$ be the first vertex from left such that there exists an edge $(v_\ell,v_{\ell'})$ such that $\ell'<\ell$. But then there exists a cycle $v_{\ell'},\dots, v_\ell, v_{\ell'}$ of length at most three and we get a contradiction from claim\ref{order} part $2$.
\end{proof}

%Next using graph $\overline{H}$ and set $G$ we define set ${V(\overline{H})}_{G}=V(\overline{H})\cap G$. 
Recall, $G$ is the set of good symbols (vertices). As a corollary of Claim~\ref{order} we have the following.

\begin{corollary}
\label{cor:order}
	$|G\setminus V(\overline{H})|\le \frac{4\alpha}{1-4\alpha}|\overline{G}|$.
\end{corollary}
\begin{proof}
	
	By Claim~\ref{order}, any cycle that we remove has at least $\frac{1}{2\alpha}$ vertices. Moreover, among them at most two are good. So the number of bad vertices in each removed cycle is at least $\frac{1-4\alpha}{2\alpha}$. Hence, total number of good vertices we remove is at most $\frac{4\alpha}{1-4\alpha}|\overline{G}|$. 
\end{proof}

Next we consider a topological ordering of  $\overline{H}$. Using this we define an ordering $\mathcal{P}$ among the symbols of $V(\overline{H})$ as follows:   
For each $i,j\in V(\overline{H})$, $i$ precedes $j$, denoted by $i\triangleleft j$ if  $i$ occurs before $j$ in the topological sorted ordering. 
%Notice this partial ordering can be found trivially in time $O(|V_{\overline{H}}|^4)=O(n^4)$ We further extend this partial order to a complete order $\mathcal{C}$ as follows: For each $i,j\in\tilde{S}$ if neither $i\triangleleft j$ nor $j\triangleleft i$, add an arbitrary order $i\triangleleft j$. 
Let $\overline{x}$ be the string over the symbols of $V(\overline{H})$ obeying the ordering of $\mathcal{P}$. Note, $V(\overline{H})$ may not contain all the $n$ vertices (or symbols). Create a permutation $\tilde{x}$ over $[n]$ by appending the symbols of $[n]\setminus V(\overline{H})$ at the end of string $\overline{x}$ in any arbitrary order. 
We claim the following.

\begin{lemma}
	\label{case1}
	$d(\tilde{x},\med)\le \frac{1}{1-4\alpha}|\overline{G}|$.
\end{lemma}

Before proving the lemma, we make the following claim.

\begin{claim}
	\label{sorder}
	For any pair of symbols $i,j\in G \cap V(\overline{H})$, if $i\triangleleft j$, then $i$ precedes $j$ in $\med$; otherwise $j$ precedes $i$ in $\med$.
\end{claim}

\begin{proof}
	For any pair of symbols $i,j\in G \cap V(\overline{H})$, as both the symbols $i,j\in G$, by Observation~\ref{diedge} there exists an edge between $i$ and $j$ in $\overline{H}$. So if $i\triangleleft j$, then there must exist a directed edge from $i$ to $j$, and therefore in at least $(1-2\alpha)m$ strings $i$ appears before $j$. As both $i$ and $j$ are aligned together in at least $(1-2\alpha)m$ strings and $1-2\alpha>2\alpha$ (for $\alpha\le 1/10$), $i$ precedes $j$ in $\med$. We can prove the other direction in a similar way. 
\end{proof}

\begin{proof}[Proof of Lemma~\ref{case1}]
	Following Claim~\ref{sorder}, between $\tilde{x}$ and $\med$ there exists a common subsequence of length at least $|G \cap V(\overline{H})|$. Hence 
\begin{align*}
d(\tilde{x},\med) &\le n-|G \cap V(\overline{H})|\\
& =  |G|+|\overline{G}|-|G \cap V(\overline{H})|\\
& =  |\overline{G}|+|G \setminus V(\overline{H})|\\
& \le \frac{1}{1-4\alpha}|\overline{G}|&\text{by Corollary~\ref{cor:order}}.
\end{align*}		
\end{proof}

\begin{proof}[Proof of Lemma~\ref{lem:low-regime-ordering}]
	Procedure ${\BRO}(S,\alpha)$ %{\BRO($S,\alpha$)} 
	outputs a string $\tilde{x}$ such that $d(\tilde{x},\med)\le \frac{1}{1-4\alpha}|\overline{G}|$. Hence by triangle inequality, 
	
	\begin{align*}
	%\label{eq:case1}
	\sum_{y\in S}d(\tilde{x},y) & \le \sum_{y\in S}\Big(d(y,\med)+d(\med,\tilde{x})\Big) \nonumber \\
	&\le \opt+\frac{m}{1-4\alpha}|\overline{G}| &\text{by Lemma~\ref{case1}} \nonumber \\
	&\le \opt+\frac{\beta}{1-4\alpha}\opt &\text{as $|\overline{G}|\le \beta \frac{\opt}{m}$} \nonumber \\
	&\le (1+\beta(1+8\alpha))\opt &\text{as $\alpha\le 1/10$}.
	\end{align*}
\end{proof}

\begin{remark}
\label{rem:time-improve}
We can improve the running time of Procedure {\BRO} from $O(n^2 m + n^4)$ to $O(n^2 m + n^3)$ by slightly modifying it, without losing much on the approximation guarantee. Currently, Procedure {\BRO} runs a while loop until there is no cycle in the graph $H$, and at each iteration computes a shortest cycle on the whole graph and delete all the vertices of that cycle (with all their incident edges). Instead of this while loop, we can enumerate over all the vertices and while enumerating a vertex $v$ compute a shortest cycle that contains $v$ and then delete all its vertices (with all their incident edges). Now each iteration takes only $O(n^2)$ time, and so the enumeration over all the vertices takes $O(n^3)$ time. Hence, the overall running time is $O(n^2 m + n^3)$.

The issue with this modification is that Claim~\ref{order} can get violated as each deleted cycle may contain more than two good vertices. However, we can claim that for any vertex $v$, in a shortest cycle $\mathcal{C}$ containing it, the ratio between the number of good and bad vertices is at most $3/(\frac{1}{2\alpha}-2)$. The claim follows from two observations. First, for any two good vertices $u_1, u_2$ in $\mathcal{C}$, either they are consecutive in $\mathcal{C}$, or there are at least $\frac{1}{2\alpha}-2$ bad vertices between them. To see this, take any two non-consecutive vertices $u_1, u_2$ in $\mathcal{C}$, and without loss of generality assume $u_1$ appears before $u_2$ in $\mathcal{C}$. By an argument similar to that in the proof of Claim~\ref{order}, if there are at most $\frac{1}{2\alpha} - 3$ bad vertices between $u_1,u_2$ then there must be a directed edge from $u_1$ to $u_2$ in $H$, contradicting the cycle $\mathcal{C}$ being a shortest cycle containing $v$. Our second observation is that, if three good vertices $u_1,u_2,u_3$ form a 3-length subpath $u_1 \rightarrow u_2\rightarrow u_3$ in $\mathcal{C}$ (which is a shortest cycle containing $v$), then $u_2=v$. This is because, since $u_1,u_2,u_3$ all are good vertices and there is an edge from $u_1$ to $u_2$ and from $u_2$ to $u_3$, there must be an edge also from $u_1$ to $u_3$, and therefore if $u_2 \ne v$, we will get a shorter cycle by following the edge $u_1$ to $u_3$ contradicting $\mathcal{C}$ being a shortest cycle containing $v$. Our claimed bound on the ratio of good and bad symbols now follows from these two observations. Hence, at the end of enumeration, the number of good symbols (or vertices) that got deleted is at most $\frac{6\alpha}{1-4\alpha}|\overline{G}|$ (which is slightly worse than that in Corollary~\ref{cor:order}). The rest of the argument will remain the same, and therefore finally, we will get a $(1+\beta(1+12\alpha))$-approximate median.
\end{remark}

\subsubsection*{Case 2: $|\overline{G}| > \beta \frac{\opt}{m}$}
 Recall, we take parameters $\delta, \alpha, \beta, \gamma, \xi, \eta$, the value of which will be set later.
\begin{lemma}
\label{lem:low-regime-input}
Let $\alpha \in (0,1/10]$ and $\beta \in (0,1)$.
Given $S \subseteq \sym_n$ of size $m$ such that the number of bad symbols
is $|\overline{G}| \geq \beta \frac{\opt}{m}$,
Procedure {\BI}($S$) outputs a $(2-\zeta)$-approximate median, where $\zeta=\frac{(1-\alpha/2)\alpha^5\beta^2}{2^{20}\log_{(1+\frac{3\alpha}{64})}^2 (8/3\alpha)}$.
\end{lemma}
We would like to mention that the constant $1/2^{20}$ in the above lemma is not optimal, and one can improve this significantly by optimizing various parameters. The proof of the above lemma resembles that of Lemma~\ref{lem:high-regime-ulam},
however it is much more intricate. We devote the remaining subsection to present the proof.

First, partition $S$ into the far and close permutations (from $\med$).
Let $F=\{ x\in S : d(x,\med) \ge (1+\tfrac{2\delta}{\alpha})\opt/m \}$ 
and $\overline{F}=S\setminus F$.

\begin{claim}
\label{clm:high-mass-close-small}
The set $\overline{F}$ satisfies:
\begin{align}
  % \label{itm:high1}
  & \opt_{\overline{F}}
  \ge (1-\tfrac{\alpha}{2})(1-\delta) \opt.
    \label{eq:high1}
  \\
  % \item \label{itm:high2}
  & \opt_{\overline{F}}(\overline{G})
  \ge \tfrac{\alpha m}{2}|\overline{G}|
    \ge \tfrac{\alpha \beta}{2}\opt_{\overline{F}}. 
    \label{eq:high2}
  % \item \label{itm:high3}
\end{align}
\end{claim} 

\begin{proof}
To prove~\eqref{eq:high1},
we use our assumption that $d(x,\med)\ge (1-\delta)\opt/m$ for all $x\in S$.
By an averaging argument $|\overline{F}| \ge \frac{2}{\alpha}|F|$,
as otherwise
\[
  \opt
%  = \opt_{F} + \opt_{\overline{F}}
  \ge |F| (1+\tfrac{2\delta}{\alpha}) \tfrac{\opt}{m}
      + |\overline{F}| (1-\delta) \tfrac{\opt}{m}
  = \big[ |F| + |\overline{F}| + \delta(\tfrac{2}{\alpha}|F| - |\overline{F}|)  \big] \tfrac{\opt}{m}
  > \opt .
\]
Thus $|F| \le \frac{\alpha}{2} |\overline{F}| < \frac{\alpha}{2} m$,
and we obtain as claimed
\begin{align*}
  \opt_{\overline{F}}
  \ge |\overline{F}| \cdot (1-\delta) \tfrac{\opt}{m}
  \ge (1-\tfrac{\alpha}{2})m \cdot (1-\delta) \tfrac{\opt}{m}.
\end{align*}

To prove~\eqref{eq:high2},
%note that $|F| < \alpha m/2$,
for every bad symbol $a\in \overline{G}$ we have 
$
  c_{\overline{F}}(a)
  \ge c(a) - |F|
  > \frac{\alpha}{2} m
$.
Thus
\begin{align*}
  \opt_{\overline{F}}(\overline{G})
  = \sum_{a\in \overline{G}} c_{\overline{F}}(a) 
  \ge \tfrac{\alpha}{2}m \cdot |\overline{G}|.
\end{align*}
Now using that $|\overline{G}| > \beta \frac{\opt}{m}$,
we get as claimed that
$
  \opt_{\overline{F}}(\overline{G})
  \ge \frac{\alpha \beta}{2} \opt 
  \ge \frac{\alpha \beta}{2}\opt_{\overline{F}}
$. 
\end{proof}

Next, consider the following subset of $\overline{F}$.
$$R:=\{x\in \overline{F} : |I_x(\overline{G})| \ge (1-\gamma)\tfrac{\alpha}{2}|\overline{G}|\}.$$
By definition of $R$, 
\begin{align*}
  \opt_{\overline{F}\setminus R}(\overline{G})
  &\le (1-\gamma)\frac{\alpha}{2}|\overline{G}|\cdot |\overline{F}\setminus R|
  \\
  &\le (1-\gamma)\opt_{\overline{F}}(\overline{G})
  & \text{by~\eqref{eq:high2}} ,
\end{align*}
and we conclude that
%\rnote{I changed this claim into a disp eqn} 
\begin{equation}
  \label{eq:high-mass-subset-small}
  \opt_R(\overline{G}) \ge \gamma \opt_{\overline{F}}(\overline{G}).
\end{equation}

\iffalse
\begin{claim}
\label{clm:high-mass-further-subset}
For any $\alpha,\gamma,\eta \in (0,1)$ and $r=\lceil \log_{1+\eta}(\frac{2}{(1-\gamma)\alpha})\rceil$, 
 there exists a subset $R^* \subseteq R$ such that
\begin{enumerate}
\item \label{itm:high-sub1} $R^* = \{x \in R  :  (1+\eta)^{i^*-1}(1-\gamma)\frac{\alpha}{2}|\overline{G}| \le |I_{x}(\overline{G})| \le (1+\eta)^{i^*}(1-\gamma)\frac{\alpha}{2}|\overline{G}|\}$ for some $i^* \in [r]$.
\item \label{itm:high-sub2} $\opt_{R^{*}}(\overline{G}) \ge \frac{\gamma}{r}\opt_{\overline{F}}(\overline{G})$.
\end{enumerate}
\end{claim}
\fi

Next partition $R$ into $r=\Big\lceil\log_{1+\eta}(\frac{2}{(1-\gamma)\alpha})\Big \rceil$ sets $R_1,R_2,\cdots,R_r$, where
\begin{align}
\label{eq:high-sub1}
R_i  = \{x \in R  :  (1+\eta)^{i-1}(1-\gamma)\frac{\alpha}{2}|\overline{G}| \le |I_{x}(\overline{G})| \le (1+\eta)^{i}(1-\gamma)\frac{\alpha}{2}|\overline{G}|\}.
\end{align}
By a straightforward averaging, there exists $R^*=R_{i^*}$ for some $i^* \in [r]$, such that
\begin{align}
\label{eq:high-sub2}
  \opt_{R^{*}}(\overline{G})
  &\ge \frac{1}{r}\opt_R(\overline{G})\nonumber
  \\
  &\ge \frac{\gamma}{r}\opt_{\overline{F}}(\overline{G})
  & \text{by~\eqref{eq:high-mass-subset-small}} .
\end{align}

 Now we further partition the set $R^{*}$ using the following procedure. Initialize a set $C = \emptyset$, and set
 \begin{align}
\label{eq:xi}
\xi=\frac{1}{2}\Big((1+\eta)^{i^*-1}(1-\gamma)\frac{\alpha}{2}\Big)^2.
\end{align} 
Then iterate over all $x \in R^{*}$ in the non-decreasing order of the value of $|I_x(G)|$. For each $x \in R^{*}$, if 
$$\forall y \in C,\; |I_x (\overline{G}) \cap I_y(\overline{G})| < \xi |\overline{G}|,$$ 
add $x$ into the set $C$, and create its "buddies set" $B_x=\emptyset$. Otherwise pick some $y \in C$ that violates the above, breaking ties arbitrarily, and add $x$ into its buddies set $B_y$.

Note, the above partitioning is solely for the sake of analysis. Since we process $R^{*}$ in the sorted order, clearly
\begin{align}
\label{eq:small-center-good}
\forall y \in C, x \in B_y,\; |I_y(G)| \le |I_x(G)|.
\end{align}

We use Lemma~\ref{lem:intersect-family} to get an upper bound on the size of the set $C$.
\begin{claim}
\label{clm:cluster-points-small}
For any $\alpha \in (0,1)$ and $\gamma \in (0,1/2)$, $|C| \le 8/\alpha$.
\end{claim}
\begin{proof}
Observe, for all pairs of distinct permutations $x,x'\in C$, $|I_x (\overline{G}) \cap I_{x'}(\overline{G})| < \xi |\overline{G}|$. Also by ~\eqref{eq:high-sub1}, for all $x \in R^{*}$ $|I_x(\overline{G})| \ge (1+\eta)^{i^*-1}(1-\gamma)\frac{\alpha}{2}|\overline{G}|$. As a corollary of Lemma~\ref{lem:intersect-family} we get, $|C| \le 2 \lceil \frac{1}{(1+\eta)^{i^*-1}(1-\gamma)\frac{\alpha}{2}} \rceil \le \frac{8}{\alpha}$ (for any $\gamma < 1/2$).
\end{proof}

Since $R^* = \bigcup_{y \in C}B_y$, by an averaging argument
\begin{align}
\label{eq:large-cluster1}
\exists y\in C,\; \opt_{B_y}(\overline{G}) & \ge \frac{\opt_{R^{*}}(\overline{G})}{|C|} \nonumber \\
&\ge \frac{\alpha\gamma}{8r} \opt_{\overline{F}}(\overline{G})\qquad\text{by ~\eqref{eq:high-sub2}}.
\end{align}

Also we get that
\begin{align}
\label{eq:interim-large-cluster}
\opt_{B_y} &\ge \opt_{B_y}(\overline{G})\nonumber\\
& \ge \frac{\alpha \gamma}{8r} \opt_{\overline{F}}(\overline{G})\qquad\;\;\text{by~\eqref{eq:large-cluster1}}\nonumber\\
& \ge \frac{\alpha^2 \beta \gamma}{16r}\opt_{\overline{F}}\qquad\quad\text{by~\eqref{eq:high2}} .
\end{align}
Clearly, $\opt_{B_y} \le \opt_{\overline{F}}$, and $\opt_{B_y}=\opt_{B_y}(G)+\opt_{B_y}(\overline{G})$. Hence we can write
\begin{align}
\label{eq:large-cluster2}
\opt_{B_y}(G) \le \Big(1-\frac{\alpha^2\beta\gamma}{16r}\Big)\opt_{B_y}.
\end{align}

\begin{claim}
\label{clm:small-regime-good-buddy}
Suppose $y$ satisfies~\eqref{eq:large-cluster1}. Then its distance to every $x \in S$ in bounded by:
\begin{align}
  \forall x \in F, \; &d(x,y) \le 2 d(x,\med).
    \label{eq:distance-noncluster1-low}
  \\
  \forall x \in \overline{F}, \; &d(x,y) \le (2+\delta(2/\alpha-1)) (1+2\delta)d(x,\med) &\text{for any }\delta \le 1/2
    \label{eq:distance-noncluster2-low}
  \\
  \forall x \in B_y, \; &d(x,y) \le (2-\rho)d(x,\med)
         & \text{where }\rho=\frac{\alpha^3 \beta \gamma (1-\gamma)}{2^8 r}.
   \label{eq:distance-cluster-low}
\end{align}
\end{claim}
\begin{proof}
To prove~\eqref{eq:distance-noncluster1-low}, consider $x \in F$. Since $y \in C \subseteq \overline{F}$, $d(y,\med) \le d(x,\med)$, and thus by the triangle inequality, $d(x,y)\le d(x,\med)+d(y,\med)\le 2 d(x,\med)$.

To prove~\eqref{eq:distance-noncluster1-low}, consider $x \in \overline{F}$. Since $d(y,\med) \le (1+2\delta/\alpha)\opt/m $, using our assumption~\eqref{eq:noneclose-low}, we get that $d(y,\med) \le \frac{1+2\delta/\alpha}{1-\delta}d(x,\med)$. So by the triangle inequality, 
\begin{align*}
d(x,y)&\le d(x,\med)+d(y,\med)\\
&\le (2+\delta(2/\alpha-1)) (1+2\delta) d(x,\med) &\text{for any }\delta \le 1/2.
\end{align*}

To prove~\eqref{eq:distance-cluster-low}, consider $x \in B_y$. Observe, $d(x,y) \le |I_x|+|I_y|-|I_x(\overline{G})\cap I_y(\overline{G})|$. Recall, by the definition, for any $x \in S$, $|I_x|=|I_x(G)|+|I_x(\overline{G})|$. Set $\eta=\frac{\xi |\overline{G}|}{2|I_x(\overline{G})|}$, and let 
\begin{align}
\label{eq:nu}
\nu=(1+\eta)^{2i^*-3}(1-\gamma)^2 (\alpha/2)^2.
\end{align}
We get that

\begin{align}
\label{eq:interism-dist-bound}
\sum_{x \in B_y}d(x,y) &\le \sum_{x \in B_y}(|I_x(G)|+|I_y(G)|) + \sum_{x \in B_y}(|I_x(\overline{G})|+|I_y(\overline{G})|-|I_x(\overline{G})\cap I_y(\overline{G})|)\nonumber\\
&\le 2 \sum |I_x(G)| + \sum(|I_x(\overline{G})|+|I_y(\overline{G})|-|I_x(\overline{G})\cap I_y(\overline{G})|)&\text{by ~\eqref{eq:small-center-good}}\nonumber\\
&\le 2 \sum |I_x(G)| + \sum (|I_x(\overline{G})|+(1+\eta)|I_x(\overline{G})| - \xi |\overline{G}|)&\text{since }x, y \in R^{*}\nonumber\\
&\le 2\sum |I_x| - \sum (\xi |\overline{G}| - \eta |I_x(\overline{G})|)\nonumber\\
& \le 2 \sum |I_x| - \nu \sum |I_x (\overline{G})|
\end{align}
where $\nu = \frac{(1+\eta)^{i^*-2} (1-\gamma) \alpha}{4} -\eta$. The last inequality follows from~\eqref{eq:high-sub1} (since $x \in R^*$) and by replacing the value of $\xi$ as in~\eqref{eq:xi}. Set
\begin{align}
\label{eq:eta-value}
\eta = (1-\gamma)\alpha/16.
\end{align}
We get
\begin{align}
\label{eq:bound-nu}
\nu & = \frac{(1+\eta)^{i^*-2} (1-\gamma) \alpha}{4} -\eta\nonumber\\
& \ge \frac{(1-\gamma)\alpha}{4(1+\eta)} -\eta\nonumber\\
& \ge \frac{(1-\gamma)\alpha}{16}&\text{from~\eqref{eq:eta-value}}.
\end{align}
Note, in the last inequality we also use $\eta \le 1$. Indeed, later we will set $\alpha , \gamma$ in such a way so that $\eta < 1/8$.

Now from~\eqref{eq:interism-dist-bound},
\begin{align*}
\sum_{x \in B_y}d(x,y) &\le 2 \sum |I_x| - \nu \sum |I_x (\overline{G})|\nonumber \\
&\le (2-\nu)\sum |I_x| + \nu \sum |I_x(G)|\nonumber\\
&\le (2-\nu)\opt_{B_y} + \nu \opt_{B_y}(G)&\text{by the definition of $B_y$ and $B_y(G)$}\nonumber\\
&\le \Big(2-\frac{\nu\alpha^2\beta \gamma}{16r}\Big)\opt_{B_y}&\text{by ~\eqref{eq:large-cluster2}}\nonumber\\
&\le \Big(2-\frac{\alpha^3 \beta \gamma (1-\gamma)}{2^8 r}\Big)\opt_{B_y}&\text{by~\eqref{eq:bound-nu}}.
\end{align*}
\end{proof}

Now we complete the proof of Lemma~\ref{lem:low-regime-input}. Let $y \in C$ be as in Claim~\ref{clm:small-regime-good-buddy} (that satisfies~\eqref{eq:large-cluster1}).
\begin{align}
\label{eq:pre-final-bound}
\sum_{x \in S}d(x,y) &\le \sum_{x \in F}d(x,y) + \sum_{x \in \overline{F}\setminus B_y} d(x,y) + \sum_{x \in B_y}d(x,y) \nonumber\\
&\le 2 \opt_F + (2+\delta(2/\alpha-1))(1+2 \delta) \opt_{\overline{F}\setminus B_y} + (2-\rho)\opt_{B_y} &\text{by Claim~\ref{clm:small-regime-good-buddy}} \nonumber\\
&= 2 \opt + (3+2/\alpha+2\delta(2/\alpha+1))\delta \opt_{\overline{F}\setminus B_y} - \rho \opt_{B_y}\nonumber\\
& \le 2 \opt + (3+2/\alpha+2\delta(2/\alpha+1))\delta \opt_{\overline{F}} - \rho \opt_{B_y}\nonumber\\
&\le 2 \opt - \Big(\frac{\rho \alpha^2\beta \gamma}{16r} - (3+2/\alpha+2\delta(2/\alpha+1))\delta\Big)\opt_{\overline{F}}&\text{by ~\eqref{eq:interim-large-cluster}}\nonumber\\
&\le \Big(2 - (1-\alpha/2)(1-\delta)\Big(\frac{\rho \alpha^2\beta \gamma}{16r} - (3+2/\alpha+2\delta(2/\alpha+1))\delta\Big)\Big) \opt&\text{by~\eqref{eq:high1}}\nonumber\\
&\le \Big(2 - (1-\alpha/2)(1-\delta)\Big(\frac{\alpha^5\beta^2 \gamma^2(1-\gamma)}{2^{12}r^2} - (3+2/\alpha+2\delta(2/\alpha+1))\delta\Big)\Big) \opt.
\end{align}
Recall, $r=\Big\lceil\log_{1+\eta}(\frac{2}{(1-\gamma)\alpha})\Big \rceil$. We have already set $\eta = (1-\gamma)\alpha/16$ in~\eqref{eq:eta-value}. Next, set $\gamma = 1/4 $ (note, this setting satisfies the requirement in Claim~\ref{clm:cluster-points-small}). Further, set 
\begin{align}
\label{eq:delta-value}
\delta = \frac{\alpha^6\beta^2}{2^{19}\log_{(1+\frac{3\alpha}{64})}^2 (8/3\alpha)}.
\end{align}
Then by simplifying~\eqref{eq:pre-final-bound}, we get
$$\sum_{x \in S}d(x,y) \le \Big(2 - \frac{(1-\alpha/2)\alpha^5\beta^2}{2^{20}\log_{(1+\frac{3\alpha}{64})}^2 (8/3\alpha)}\Big).$$
This concludes the proof of Lemma~\ref{lem:low-regime-input}.

\paragraph*{Proof of Theorem~\ref{thm:worst-approx-poly}. }Recall, if our input set $S$ violates assumption~\eqref{eq:noneclose-low}, then we get a $(2-\delta)$-approximate median using Procedure {\BI}. Also recall the setting of parameter $\delta$ in~\eqref{eq:delta-value}. Next, set $\alpha=1/10$ and $\beta =1/2$. Now Theorem~\ref{thm:worst-approx-poly} follows from Lemma~\ref{lem:low-regime-ordering} and~\ref{lem:low-regime-input}.

\subsection{Generalization to Edit Distance (for the High Regime)} 
\label{sec:high-edit}

So far, all our results are only for the Ulam metric. In this section, we will describe how to extend our result of Section~\ref{sec:high-regime} to the edit metric space, which is a generalization of the Ulam. The edit distance between two strings is defined as the minimum number of insertion, deletion and character substitution operations required to transform one string into another. For the simplicity in exposition, we start with a special variant of the edit distance, where character substitution is not allowed. (Originally, Levenshtein~\cite{levenshtein1965binary} defined both the variants, with and without the substitution operation.) In this section, we refer this special variant also as the edit distance. For any two strings $x,y$, their edit distance, denoted by $\ed(x,y)$, is the minimum number of insertion and deletion operations to transform $x$ into $y$. So $\ed(x)=|x| + |y| - |\lcs(x,y)|$.  %The variant we described here, with only insertion and deletion operations, is sometimes also referred to as the \emph{{\lcs}-distance}~\cite{} in the literature.

We now define the median under the edit distance metric,
requiring it has the same length as the input strings. 
Formally, the \emph{length-$n$ edit-median} of a set of strings $S\subseteq \Sigma^n$
is a string $\med\in \Sigma^n$ such that $\sum_{x \in S}\ed(x,\med)$ is minimized.
A $c$-approximate length-$n$ edit-median is defined analogous to that for the Ulam metric.

\worstpolyedit*

\iffalse
\begin{theorem}
\label{thm:high-regime-edit}
Given a set of strings $S \subseteq \Sigma^n$ over alphabet $\Sigma$
with $\opt(S) \ge |S|n/c$ for some $c > 1$,
Procedure {\BI($S$)} outputs a $(2-\frac{1}{50c^2})$-approximate edit-median. 
\end{theorem}
\fi

Let $\med \in \Sigma^n$ be an (arbitrary) median of $S$; then $\opt(S)=\sum_{x \in S}\ed(x,\med)$.
We use the argument used in the proof of Lemma~\ref{lem:high-regime-ulam},
but change the definition of $I_x$ for $x \in S$ as follows:
Fix an optimal alignment (or a {\lcs}) between $\med$ and $x$,
and let $I_x$ be the set of positions $i \in [n]$ such that $\med(i)$ is not aligned by this alignment. 
Notice that $|I_x|=\ed(x,\med)/2$ since $x$ and $\med$ have the same length $n$. Furthermore, for all $x \ne y \in S$, 
$$ |\lcs(x,y)| \ge |\overline{I_x}\cap \overline{I_y}|, $$
because the positions in $\overline{I_x}\cap \overline{I_y}$
define a subsequence of $\med$ that is common to both $x$ and $y$.
Thus, 
\begin{align}
\label{eq:edit-lcs}
  \ed(x,y)
  &=2(n-|\lcs(x,y)|)\nonumber\\
  &\le 2(n-|\overline{I_x}\cap \overline{I_y}|)\nonumber\\
  &=   2 |I_x\cup I_y|\nonumber\\
  &=   2(|I_x| + |I_y| - |I_x \cap I_y|)\nonumber\\
  &\le \ed(x,\med) + \ed(y,\med) - |I_x \cap I_y|.
\end{align}
Then we follow the argument as in the proof of Lemma~\ref{lem:high-regime-ulam} to identify a point $y \in \Sigma^n$ (a cluster)
as in Claim~\ref{clm:large-cluster-mass},
and bound the distance from $y$ to all $x\in S$
as in Claim~\ref{clm:distance-cluster}.
To prove the bound~\eqref{eq:distance-cluster}, we use~\eqref{eq:edit-lcs}, and the rest of the arguments will remain the same.

\begin{remark}
We can further extend our proof to a more generalized edit distance notion, with character substitution also as a valid edit operation. In this case, the proof will be slightly more involved (by considering different cases depending on whether the unaligned index positions are for substitutions or deletions). However, if we allow the median string to be of arbitrary length (not necessarily the same as that of input strings), our proof will fail. Indeed, in this case, there exists an input set $S$ with $\opt \ge \Omega(n |S|)$ such that Procedure {\BI($S$)} does not achieve approximation better than factor 2.
\end{remark}

%%% Local Variables:
%%% mode: latex
%%% TeX-master: "medianUlam"
%%% End:

\section{Approximate Median in a Probabilistic Model}
\label{sec:avg-case}
Consider a permutation $x \in \sym_n$. Then take a set of "noisy" copies of $x$, where each noisy copy is generated from $x$ by moving "a few" randomly chosen symbols in randomly chosen positions. Formally, for any $\epsilon \in (0,1)$ define $S(x,\epsilon,m)$ as a set of $m$ permutations $x_1,\cdots,x_m \in \sym_n$ such that for each $i\in [m]$ $x_i$ is generated from $x$ in the following way:
Select each symbol in $[n]$ independently with probability $\epsilon$. Let the set of selected symbols be $\Sigma_i$. For each symbol $a \in \Sigma_i$ choose another symbol $b_i(a)$ independently uniformly at random from $[n]$, and then move the symbol $a$ from its original position (in $x$) to right next to $b_i(a)$. Let $\Sigma_i^b =\{b_i(a)  :  a \in \Sigma_i\}$.

Denote the set of all move  operations performed to generate $x_i$ by the set of tuples $(a,b_i(a))$. Let $\Sigma_i^e = \{(a,b_i(a)) :  a\in \Sigma_i\}$. For each $i \in [m]$, define set $\Sigma^r_i = \{a \in \Sigma_i  :  b_i(a) \in \Sigma_i\}$. %Notice in this case $\Sigma^r_i$ also contains symbol $b$. 

Given $S$ drawn from $S(x,\epsilon,m)$, the objective is to find its median. Throughout this section, all the probabilities are over the randomness used to generate this set $S$. Now we state the main theorem of this section.

\avgpoly*

Next, we state a few necessary observations about permutations in $S(x,\epsilon,m)$, following the simple application of Chernoff bound.

\begin{observation}
\label{obs:size-move}
For any $\epsilon \in (0,1)$, any $n \in \mathbb{N}$, a permutation $x \in \sym_n$ and any $m \in \mathbb{N}$, let $S=S(x,\epsilon,m)$. Then the followings hold.
\begin{enumerate}
\item \label{itm:move1} For any $i\in [m]$, $\Pr[|\Sigma_i| \not \in (1\pm\frac{1}{\sqrt{\log n}})\epsilon n] \le e^{-\epsilon n/4 \log n}$.
%\item \label{itm:move2} For any $m \le e^{\epsilon n/10 \log n}$ for all $i \in [m]$, $|\Sigma_i| \in (1\pm\frac{1}{\sqrt{\log n}})\epsilon n$ with probability at least $1-\frac{1}{m}$.
\item \label{itm:move3} For any two $x_i \ne x_j \in S$, $\Pr[|\Sigma_i\cap \Sigma_j| \not \in (1\pm\frac{1}{\sqrt{\log n}}) \epsilon^2 n] \le e^{-\epsilon^2 n/4 \log n}$.
\end{enumerate}
\end{observation}
\begin{proof}
Observe, for any $i\in [m]$, $\mathbb{E}[|\Sigma_i|]=\epsilon n$. So by Chernoff bound we get Item~\ref{itm:move1}. %Then by a union bound over all $i \in [m]$ we get Item~\ref{itm:move2}.

For any two $x_i \ne x_j \in S$, $\mathbb{E}[|\Sigma_i \cap \Sigma_j|]=\epsilon^2 n$. Item~\ref{itm:move3} now follows from Chernoff bound.
\end{proof}

 Recall, for any $I\subseteq [n]$, $x(I) = \{x(i)  :  i \in I\}$.
 \begin{observation}
 \label{obs:interval-intersection}
 For any $\epsilon \in (0,1)$, any $n \in \mathbb{N}$, a permutation $x \in \sym_n$ and any $m \in \mathbb{N}$, let $S=S(x,\epsilon,m)$. Then for any permutation $x_i\in S$, probability that for all intervals $I \subseteq [n]$ of size at least $\frac{15}{\epsilon}\log n$, $|\Sigma_i \cap x(I)| \ge 2\epsilon |I|$ is at most $n^{-3}$.
 \end{observation}
 \begin{proof}
 Observe, For any $x_i\in S$ and an interval $I \subseteq [n]$, $\mathbb{E}[|\Sigma_i \cap x(I)|]=\epsilon |I|$. So by Chernoff bound
 $$\Pr[|\Sigma_i \cap x(I)| \ge 2 \epsilon |I|] \le e^{-\epsilon |I|/3}.$$
 Now the observation follows from a union bound over all intervals $I \subseteq [n]$.
 \end{proof}

A similar claim is also true for the set $\Sigma_i^b$.
\begin{observation}
\label{obs:interval-intersection-b}
For any $\epsilon \in (0,1)$, any $n \in \mathbb{N}$, a permutation $x \in \sym_n$ and any $m \in \mathbb{N}$, let $S=S(x,\epsilon,m)$. Then for any permutation $x_i\in S$, probability that for all intervals $I \subseteq [n]$ of size at least $\frac{15}{\epsilon}\log n$, $|\Sigma_i^b \cap x(I)| \ge 2\epsilon |I|$ is at most $n^{-3}$.
\end{observation}

\subsection{Hidden Permutation and Approximate Median}
To prove Theorem~\ref{thm:avg-case-poly} we design an algorithm that given a set $S$ drawn from $S(x,\epsilon,m)$, finds a "good approximation" of $x$. Recall, our main goal is to find a median permutation $\med$ for $S$. The following theorem explains why it suffices to find $x$ instead of an actual median.
\begin{theorem}
\label{thm:approx-median-hidden}
For every $\epsilon \in (0,1/12)$, any large enough $n \in \mathbb{N}$, a permutation $x \in \sym_n$, $20 \le m \le n$ and $\delta=\frac{20}{m}+ \frac{3}{\log(n/\epsilon)}$, for a set of permutations $S$ drawn from $S(x,\epsilon,m)$,
$$\Pr[\obj(S,x) \le (1+\delta)\opt(S)] \ge 1-mn^{-1.5}.$$ 
\end{theorem}

Before proving the above lemma, let us first make an observation regarding a longest common subsequence ({\lcs}) between $x$ and each $x_i$.

\begin{lemma}
\label{lem:near-unique-lcs}
For any $\epsilon \in (0,1/4)$, a large enough $n \in \mathbb{N}$, a permutation $x \in \sym_n$, and any $m \le n$, let $S=S(x,\epsilon,m)$. For any $x_i \in S$, let $L_{i}$ denote the set of symbols in an {\lcs} between $x_i$ and $x$. Then with probability at least $1-2n^{-3}$,
$$|L_{i} \cap \Sigma_i | \le \frac{30}{\epsilon n}|\Sigma_i|\log n .$$
\end{lemma}
Note, {\lcs} between two permutations may not be unique. However the above lemma is true for any {\lcs} between $x_i$ and $x$. 

\begin{proof}
Consider a symbol $a \in \Sigma_i $. So $(a,b_i(a)) \in \Sigma_i^e$. Consider $k_a,k_b \in [n]$ such that $x(k_a)=a$ and $x(k_b)=b_i(a)$. Let us consider the interval $I_i(a)=\{\min,\min+1,\cdots,\max\}$ where $\min=\min\{k_a, k_b\}$ and $\max=\max\{k_a, k_b\}$. First we claim that for all $a \in \Sigma_i$ if $a \in L_{i}$ then $|I_i(a)| < \frac{15}{\epsilon}\log n$ with probability at least $1-2n^{-3}$.

For the contradiction sake, let us assume that for some $a \in \Sigma_i$, $|I_i(a)| \ge \frac{15}{\epsilon} \log n$. Observe, if any {\lcs} between $x_i$ and $x$ contains $a$ then it cannot contain any other symbol from the set $x(I_i(a))\setminus \Sigma_i$. Since $|I_i(a)| \ge \frac{15}{\epsilon}\log n$, by Observation~\ref{obs:interval-intersection} with probability at least $1-n^{-3}$, $|\Sigma_i \cap x(I_i(a))| \le 2\epsilon |I_i(a)|$. Thus $|x(I_i(a))\setminus \Sigma_i| \ge 2$ for any $\epsilon \le 1/4$ and large enough $n$. So if we exclude $a$ from the common subsequence and include all the symbols from the set $x(I_i(a))\setminus \Sigma_i $, we get another common subsequence of longer length. Hence $a \not \in L_{i}$.

Let us now consider the set $R = \{a \in \Sigma_i  :  |I_i(a)| < \frac{15}{\epsilon} \log n\}$. By the above claim we get that with probability at least $1-2n^{-3}$,
\begin{align}
\label{eq:near-unique-lcs}
|L_i \cap \Sigma_i| \le |R|.
\end{align}
 
 Since by the definition of $S(x,\epsilon,m)$, for any $a \in \Sigma_i$, $b_i(a)$ is chosen independently uniformly at random from $[n]$,
$$\Pr[a \in R] \le \frac{15}{\epsilon n} \log n.$$
So by linearity of expectation, $\mathbb{E}[|R|] \le |\Sigma_i | \cdot \frac{15}{\epsilon n}\log n $. Note, a symbol $a\in R$ depending only on the choice of $b_i(a)$, and thus independent of other symbols being in $R$. So by Chernoff bound, 
$$\Pr\Big[|R| \le \frac{30}{\epsilon n}|\Sigma_i |\log n \Big] \ge 1-e^{-\frac{4\log n}{\epsilon n}|\Sigma_i|}.$$
The lemma now follows from~\eqref{eq:near-unique-lcs} and the bound of $|\Sigma_i|$ established in Observation~\ref{obs:size-move}.
\end{proof}

A similar claim is also true for \emph{any} {\lcs} between any two distinct $x_i$ and $x_j$, the proof of which is also similar to the above.
\begin{lemma}
\label{lem:perm-perm-lcs}
For any $\epsilon \in (0,1/5)$, a large enough $n \in \mathbb{N}$, a permutation $x \in \sym_n$, and any $m \le n$, let $S=S(x,\epsilon,m)$. For any two $x_i \ne x_j \in S$, let $L_{i,j}$ denote the set of symbols in an {\lcs} between $x_i$ and $x_j$. Then with probability at least $1-3n^{-3}$,
$$\Big|L_{i,j} \bigcap (\Sigma_i \cup \Sigma_j)\Big| \le \frac{15}{\epsilon^2 n}|\Sigma_i \cup \Sigma_j|\log n .$$
\end{lemma}
\begin{proof}
Consider a symbol $a \in \Sigma_i \setminus \Sigma_j $. So $(a,b_i(a)) \in \Sigma_i^e$. Consider $k_a,k_{b_i(a)} \in [n]$ such that $x(k_a)=a$ and $x(k_{b_i(a)})=b_i(a)$. Let us consider the interval $I_i(a)=\{\min,\min+1,\cdots,\max\}$ where $\min=\min\{k_a, k_{b_i(a)}\}$ and $\max=\max\{k_a, k_{b_i(a)}\}$. First we claim that for all $a \in \Sigma_i \setminus \Sigma_j$ if $a \in L_{i,j}$ then $|I_i(a)| < \frac{15}{\epsilon}\log n$ with probability at least $1-2n^{-3}$.

For the contradiction sake, let us assume that for some $a \in \Sigma_i \setminus \Sigma_j$, $|I_i(a)| \ge \frac{15}{\epsilon} \log n$. Observe, if any {\lcs} between $x_i$ and $x_j$ contains $a$ then it cannot contain any other symbol from the set $x(I_i(a))\setminus (\Sigma_i\cup \Sigma_j)$. Since $|I_i(a)| \ge \frac{15}{\epsilon}\log n$, by Observation~\ref{obs:interval-intersection} with probability at least $1-2n^{-3}$, $|(\Sigma_i\cup \Sigma_j ) \cap x(I_i(a))| \le 4\epsilon |I_i(a)|$. Thus $|x(I_i(a))\setminus (\Sigma_i \cup \Sigma_j)| \ge 2$ for any $\epsilon \le 1/5$ and large enough $n$. So if we exclude $a$ from the common subsequence and include all the symbols from the set $x(I_i(a))\setminus (\Sigma_i \cup \Sigma_j) $, we get another common subsequence of longer length. Hence $a \not \in L_{i,j}$.

Similarly, for all the symbols $a \in \Sigma_j \setminus \Sigma_i$, if $a \in L_{i,j}$ then $|I_i(a)| < \frac{15}{\epsilon}\log n$ with probability at least $1-2n^{-3}$ (where $I_j(a)$ is defined analogous to $I_i(a)$).

Next consider a symbol $a \in \Sigma_i \cap \Sigma_j$. So $(a,b_i(a)) \in \Sigma_i^e$ and $(a,b_j(a)) \in \Sigma_j^e$. Consider $k_{b_i(a)},k_{b_j(a)} \in [n]$ such that $x(k_{b_i(a)})=b_i(a)$ and $x(k_{b_j(a)})=b_j(a)$. Let us consider the following interval 
$$I_{i,j}(a)=\{\min\{k_{b_i(a)},k_{b_j(a)}\}, \cdots, \max\{k_{b_i(a)},k_{b_j(a)}\} \}.$$
By using an argument similar to that for $I_i(a)$, we claim that for all $a \in \Sigma_i \cap \Sigma_j$ if $a \in L_{i,j}$ then $|I_{i,j}(a)| < \frac{15}{\epsilon} \log n$ with probability at least $1-2n^{-3}$.

Let us now consider the following three sets:
\begin{align*}
R_i &= \{a \in \Sigma_i \setminus \Sigma_j  :  |I_i(a)| < \frac{15}{\epsilon} \log n\}\\
R_j &= \{a \in \Sigma_j \setminus \Sigma_i  :  |I_j(a)| < \frac{15}{\epsilon} \log n\}\\
R_{i,j} &= \{a \in \Sigma_i \cap \Sigma_j  :  |I_{i,j}(a)| < \frac{15}{\epsilon} \log n\}
\end{align*}
 By the argument so far we get that with probability at least $1-2n^{-3}$, 
 \begin{align}
 \label{eq:perm-perm-lcs}
 |L_i \cap (\Sigma_i\cup \Sigma_j)| \le |R_i| + |R_j| + |R_{i,j}|.
 \end{align}
 Since by the definition of $S(x,\epsilon,m)$, for any $a \in \Sigma_i$, $b_i(a)$ is chosen independently uniformly at random from $[n]$,
$$\Pr[a \in R_i] \le \frac{15}{\epsilon n} \log n.$$
So by linearity of expectation, $\mathbb{E}[|R_i|] \le |\Sigma_i \setminus \Sigma_j | \cdot \frac{15}{\epsilon n}\log n $. Note, a symbol $a\in R_i$ depending only on the choice of $b_i(a)$, and thus independent of other symbols being in $R_i$. So by Chernoff bound, 
$$\Pr\Big[|R_i| \le \frac{30}{\epsilon n}|\Sigma_i \setminus \Sigma_j|\log n \Big] \ge 1-e^{-\frac{4\log n}{\epsilon n}|\Sigma_i\setminus \Sigma_j|}.$$
 
 Similarly we get
\begin{align*}
\Pr\Big[|R_j| \le \frac{30}{\epsilon n}|\Sigma_j \setminus \Sigma_i|\log n \Big] &\ge 1-e^{-\frac{4\log n}{\epsilon n}|\Sigma_j\setminus \Sigma_i|}\\
\Pr\Big[|R_{i,j}| \le \frac{15}{\epsilon^2 n}|\Sigma_i \cap \Sigma_j|\log n \Big] &\ge 1-e^{-\frac{4(1/\epsilon - 1)\log n}{\epsilon n}|\Sigma_i\cap \Sigma_j|}.
\end{align*} 
The lemma now follows from~\eqref{eq:perm-perm-lcs} and the bounds of $|\Sigma_i|, |\Sigma_j|, |\Sigma_i \cap \Sigma_j|$ established in Observation~\ref{obs:size-move}.
\end{proof}

As a corollary of the above lemma we get the following.
\begin{corollary}
\label{cor:perm-perm-dist}
For any $\epsilon \in (0,1/5]$, a large enough $n \in \mathbb{N}$, a permutation $x \in \sym_n$, and any $m \le n$, let $S=S(x,\epsilon,m)$. For any two $x_i \ne x_j \in S$, with probability at least $1-3n^{-3}$,
$$d(x_i,x_j)\ge \Big(1-\frac{15}{\epsilon^2 n}\log n\Big)|\Sigma_i \cup \Sigma_j| .$$
\end{corollary}
\begin{proof}
For any two $x_i,x_j$, let $L_{i,j}$ be the set of symbols in an {\lcs} between them. Then
\begin{align*}
d(x_i,x_j) = |\overline{L_{i,j}}| & \ge |\overline{L_{i,j}} \cap (\Sigma_i \cup \Sigma_j)|\\
& = |\Sigma_i \cup \Sigma_j| - |L_{i,j} \cap (\Sigma_i \cup \Sigma_j)|\\
& \ge \Big(1-\frac{15}{\epsilon^2 n} \log n\Big)|\Sigma_i \cup \Sigma_j|
\end{align*}
where the last inequality follows from Lemma~\ref{lem:perm-perm-lcs}.
\end{proof}

\paragraph*{Basics of Information Theory. }
To prove Theorem~\ref{thm:approx-median-hidden} we use information-theoretic (encoding-decoding) argument. Before proceeding with the details of the proof, let us first recall a few basic definitions and notations from information theory. For further exposition, readers may refer to any standard textbook on information theory (e.g., ~\cite{CT06}).

Let $X,Y$ be discrete random variables on a common probability space. The \emph{entropy} of $X$ is defined as $H(X):= - \sum_x \Pr[X=x] \log (\Pr[X=x])$.  The \emph{joint entropy} of $(X,Y)$ is defined as $H(X,Y) := - \sum_{(x,y)} \Pr[X=x,Y=y] \log (\Pr[X=x,Y=y])$. The \emph{conditional entropy} of $Y$ given $X$ is defined as $H(Y \mid X) := H(Y) - \sum_{(x,y)}{\Pr[X=x,Y=y] \log \frac{\Pr[X=x,Y=y]}{\Pr[X=x]\Pr[Y=y]}}$.

\begin{proposition}[Chain Rule of Entropy]
\label{prop:chnentro} 
Let $X,Y$ be discrete random variables on a common probability space. Then $H(X,Y)=H(X)+H(Y \mid X)$.
\end{proposition}

The seminal work of Shannon~\cite{Sha48} establishes a connection between the entropy and the expected length of an optimal code that encodes a random variable.
\begin{theorem}[Shannon's Source Coding Theorem~\cite{Sha48}]
\label{thm:source-coding}
Let $X$ be a discrete random variable over domain ${\cal X}$. Then for every uniquely decodable code $C :{\cal X} \to \{0,1\}^*$,  $\mathbb{E}(|C(X)|) \ge H(X)$. Moreover, there exists a uniquely decodable code $C:{\cal X} \to \{0,1\}^*$ such that $\mathbb{E}(|C(X)|) \le H(X)+1$.
\end{theorem}

\paragraph*{Proof of Theorem~\ref{thm:approx-median-hidden}. }
Our proof will go via an information-theoretic (encoding-decoding based) argument. First, we will argue that one can encode the set $S$ by specifying the move operations to produce $x_i$'s from a median $\med$. Then we will show that given $x$, using $x_i$'s and extra "few" bits, one can decode all the random move operations of $\Sigma_i^e$'s. Now using Shannon's source coding theorem we will get a lower bound on the optimum median objective value $ \opt(S) = \sum_{x_i \in S}d(\med,x_i)$. Then compare that with the value obtained by $x$, i.e., $\obj(S,x) = \sum_{x_i \in S}d(x,x_i)$ to get the claimed approximation guarantee. 

We formalize the above argument below. It is not hard to derive the following.
\begin{claim}
\label{clm:entropy}
$H(\Sigma_1^e,\cdots,\Sigma_m^e|x)= \epsilon n m \log (\frac{n}{\epsilon})$.
\end{claim}
\begin{proof}
Here we show that $H(\Sigma_i^e|x) \ge \epsilon n \log (\frac{n}{\epsilon})$. Observe, by the definition of $S(x,\epsilon,m)$, given any $x$, any pair of symbols $(a,b) \in [n] \times [n]$ is in $\Sigma_i^e$ with probability $\epsilon/n$. Hence by the definition of entropy
$$H(\Sigma_i^e|x) = \sum_{(a,b) \in [n] \times [n]}\frac{\epsilon}{n} \log (n/\epsilon) = \epsilon n \log (n/\epsilon).$$
Since $\Sigma_i^e$'s are mutually independent (given $x$), the claim follows from the chain rule of entropy (Proposition~\ref{prop:chnentro}).
\end{proof}

For the sake of contradiction let us assume that for $\delta = \frac{20}{m}+ \frac{3}{\log(n/\epsilon)}$
\begin{align*}
\obj(S,x) > (1+\delta)\opt(S).
\end{align*}

Now we encode the random variables $\Sigma_1^e,\cdots,\Sigma_m^e$ given $x$ in the following way: First, encode the set of move operations to transform $x$ into $\med$. Then encode the set of move operations from $\med$ to each of $x_i$.

From the above two information, we can decode all the $x_i$'s. However, that is not sufficient since we would like to decode back $\Sigma_1^e,\cdots,\Sigma_m^e$ given $x$. Now for each $i \in [m]$, compute an {\lcs} between $x$ and $x_i$, and let $L_i$ denote the set of symbols in the computed {\lcs}. (Note, to make sure that the encoder and the decoder compute the same set $L_i$, we consider the {\lcs} computed by a fixed deterministic algorithm.)

For each $i \in [m]$ we also encode the set $\overline{L_i} \Delta \Sigma_i$ (where $\Delta$ denotes the symmetric difference) so that using this information and $L_i$, decoder can decode the set $\Sigma_i$. Once a decoder identifies the set $\Sigma_i$, the next task for the decoder is to identify $b_i(a)$ for each $a \in \Sigma_i$, and thus all the tuples $(a,b_i(a))\in \Sigma_i^e$. Now for each $a \in \Sigma_i$ consider the symbol $\hat{b}_i(a)$ that precedes $a$ in $x_i$, i.e., if $x_i(k)=a$ then $\hat{b}_i(a)=x_i(k-1)$ (note, the arithmetic on indices is under modulo $n$). Let 
$$C_i = \{a \in \Sigma_i  :  \hat{b}_i(a)\ne b_i(a)\}.$$

Let $k_{i,a} \in [n]$ be such that $x_i(k_{i,a})=a$. Define the set 
$$J_i(a) : = \Big\{k_{i,a}-\frac{15}{\epsilon}\log n, \cdots, k_{i,a}+\frac{15}{\epsilon}\log n\Big\}.$$ 
Next consider an interval $R_i(a) \subseteq [n]$ of size $\frac{30}{\epsilon} \log n$, such that $x(R_i(a))$ has the maximum intersection with $x_i(J_i(a))$, breaking ties arbitrarily but in a fixed manner (to make $R_i(a)$ unique).
\begin{claim}
\label{clm:decode-bounded}
For any $a \in \Sigma_i$, with probability at least $1-3n^{-3}$, $b_i(a) \in x(R_i(a))$.
\end{claim}
We defer the proof of the claim to the latter part of this section. Now assuming the above claim, by a union bound, for all symbols $a \in C_i$ we need extra $\lceil \log |R_i(a)| \rceil$ bits for each to specify the symbol $b_i(a)$ with probability at least $1-3n^{-2}$. 

Let us now describe the whole encoding.
\begin{enumerate}
\item Encode the set of move operations from $x$ to $\med$. Denote this part by $E_1$.
\item Encode the set of move operations from $\med$ to each of $x_i$. Denote this part by $E_2$.
\item For each $i \in [m]$ encode the set $\overline{L_i} \Delta \Sigma_i$. Denote this part by $E_3$.
\item For each $a\in C_i$ specify the symbol $b_i(a)$ in the set $x(R_i(a))$. Denote this part by $E_4$.
\end{enumerate}

For any input set $S$ let us denote the output string of the above encoding algorithm by $E=E(S)$. Before describing how decoder will use $E$ to decode back $\Sigma_i^e$'s, let us bound the size (in terms of number of bits) of string $E$. Note that any $\ell$ move operations can be encoded as a subset of $n^2$ move operations (there are only $n^2$ possible move operations over a permutation) using $\log {n^2\choose \ell}$ bits. So the length of $E_1$ is at most $\log {n^2\choose d(\med,x)} \le 2 d(\med,x) \log n$. The length of $E_2$ is at most $\sum_{x_i \in S} \log {n^2 \choose d(\med,x_i)} \le \sum_{x_i \in S} d(\med,x_i)\log(en^2/d(\med,x_i))$. 

Note, for each $i \in [m]$, 
\begin{align*}
|\overline{L_i} \Delta \Sigma_i| & = |\overline{L_i}| + |\Sigma_i| - 2 |\overline{L_i} \cap \Sigma_i| \\
& = n - |L_i| + |\Sigma_i| - 2 (|\Sigma_i| - |L_i \cap \Sigma_i|)\\
& \le  2 |L_i \cap \Sigma_i| & \text{since }|L_i| \ge n - |\Sigma_i|.
\end{align*}

So by applying Lemma~\ref{lem:near-unique-lcs} we get that with probability at least $1-2mn^{-3}$,
\begin{align*}
\forall i \in [m],\;|\overline{L_i} \Delta \Sigma_i| \le \frac{60}{\epsilon n}\log n |\Sigma_i|.
\end{align*}
So the length of $E_3$ is $\sum_{i\in [m]} |\overline{L_i} \Delta \Sigma_i| \log n \le 120 m \log^2 n$ with probability at least $1-4mn^{-3}$ (by Observation~\ref{obs:size-move}). 

The length of $E_4$ is bounded by 
\begin{align*}
\sum_{i\in [m]}\sum_{a \in C_i}\lceil \log |R_i(a)| \rceil \le \sum_{i\in [m]}|C_i| \lceil \log (\frac{30}{\epsilon} \log n) \rceil
\end{align*}
because by definition for each $a \in C_i$, $|R_i(a)|= \frac{30}{\epsilon} \log n$. Observe, $\hat{b}_i(a) \ne b_i(a)$ only if either $b_i(a) \in \Sigma_i$ or there is another symbol $c\in \Sigma_i$ such that $b_i(c)=b_i(a)$. So $\mathbb{E}[|C_i|] \le 2 \epsilon |\Sigma_i| / n$. Then by Observation~\ref{obs:size-move} and Chernoff bound we get
\begin{align*}
\Pr[|C_i| \le 4\epsilon^2 n] \ge 1-e^{-\epsilon^2 n/4}.
\end{align*}
This implies that the size of $E_4$ is at most $4 \epsilon^2 nm (\log (30/\epsilon) + \log \log n + 1)$ with probability at least $1-me^{-\epsilon^2 n/4}$. So with probability at least $1-\frac{mn^{-1.5}}{4}$ (for large enough $n$) length of the total encoded string $E$ is bounded by
\begin{align}
\label{eqn:encoding}
2 d(\med,x) \log n &+ \sum_{x_i \in S} d(\med,x_i)\log(en^2/d(\med,x_i))\nonumber\\
& + 120 m \log^2 n + 4 \epsilon^2 n m (\log (30/\epsilon) + \log \log n + 1).
\end{align}

Given this $E$ and $x$ the decoding procedure works as follows:
\begin{enumerate}
\item Use $E_1$ and $E_2$ to construct $x_i$'s.
\item Compute the {\lcs} $L_i$ between $x$ and $x_i$ for each $i \in [m]$. Then use $E_3$ to get back $\Sigma_i$'s.
\item For each $i \in [m]$ and $a\in \Sigma_i\setminus C_i$ compute $\hat{b}_i(a)$; and for each $a \in C_i$ use $E_4$ to get back $b_i(a)$.
\end{enumerate}

Recall, the objective of the decoder is to get back $\Sigma^e_1,\cdots, \Sigma^e_m$, where $\Sigma^e_i=\{(a,b_i(a))  :  a\in \Sigma_i\}$. By the definition of $C_i$ for any $a\in \Sigma_i\setminus C_i$, $b_i(a)=\hat{b}_i(a)$. For all $i \in [m]$ and $a \in C_i$, by Claim~\ref{clm:decode-bounded} we get back $b_i(a)$ with probability at least $1-3mn^{-2}$ (where the probability bound follows from a union bound over all $i \in [m]$ and $a \in C_i$). So with probability at least $1-\frac{mn^{-1.5}}{4}$ (for large enough $n$) the above decoding procedure recovers the random sets $\Sigma_1^e,\cdots,\Sigma_m^e$ given $x$. 

Next consider the input set $S$ for which one of the following four conditions holds:
\begin{enumerate}
\item The decoding procedure fails,
\item Length of the encoded string $E$ is more than the bound of~(\ref{eqn:encoding}),
\item For some $i \in [m]$, $|\Sigma_i| \not \in (1\pm\frac{1}{\sqrt{\log n}})\epsilon n$,
\item For some $i \ne j \in [m]$, $d(x_i,x_j) \le \epsilon n$.
\end{enumerate}
We call such an input set \emph{bad}; otherwise call it \emph{good}.

We have already seen that each of the first two conditions holds with probability at most $mn^{-1.5}/4$. It follows from Observation~\ref{obs:size-move} and Corollary~\ref{cor:perm-perm-dist} that the last two conditions hold with probability at most $mn^{-1.5}/2$ for large enough $n$. So an input set is bad only with probability at most $p=mn^{-1.5}$.

The \emph{encoder} can check beforehand whether a given input set $S$ is bad or not (since it can simulate both the encoding and the decoding procedure). If it finds $S$ to be bad, it uses $H(\Sigma_1^e,\cdots, \Sigma_m^e|x) + 1$ bits to explicitly encode $\Sigma_1^e,\cdots, \Sigma_m^e$, and append a bit $0$ in the beginning of the encoded string. Otherwise, it uses the string generated by the previously described encoding procedure appended with a bit $1$ in the beginning. (This first bit is used by the decoder to distinguish between the above two types of inputs.)

So the expected length of the encoding is at most
\begin{align}
\label{eqn:expected-encoding}
&p\cdot\big(H(\Sigma_1^e,\cdots,\Sigma_m^e | x) + 1\big) +(1-p) \cdot \big(2 d(\med,x) \log n + \sum_{x_i \in S} d(\med,x_i)\log(en^2/d(\med,x_i)) \nonumber \\
&+120 m \log^2 n + 4 \epsilon^2 nm (\log (30/\epsilon) + \log \log n + 1)\big) + 1
\end{align}
which is at least $H(\Sigma_1^e,\cdots,\Sigma_m^e | x)$ by Shannon's source coding theorem (Theorem~\ref{thm:source-coding}). So by Claim~\ref{clm:entropy} the above expression~(\ref{eqn:expected-encoding}) is at least $\epsilon n m \log (\frac{n}{\epsilon})$. This implies the following
\begin{align}
\label{eq:source-coding}
2 d(\med,x) \log n &+ \sum_{x_i \in S} d(\med,x_i)\log(en^2/d(\med,x_i)) + 120 m \log^2 n \nonumber \\ 
& + 4 \epsilon^2 nm (\log (30/\epsilon) + \log \log n + 1) +1/(1-p) \ge \epsilon nm \log(n/\epsilon).
\end{align}

Recall, for the sake of contradiction we have assumed that
$$\sum_{x_i \in S}d(x,x_i) > (1+\delta)\sum_{x_i \in S}d(\med,x_i).$$

Now note that for any good input set there can be at most one $x_i$ such that $d(\med,x_i) \le \epsilon n /2$. Otherwise there will be $x_i,x_j \in S$ such that $d(x_i,x_j) \le \epsilon n$ by the triangle inequality (satisfying the fourth condition of an input set being bad). Now if there is an $x_i$ such that $d(\med,x_i) \le \epsilon n /2$ then we can upper bound $d(\med,x_i)\log(en^2/d(\med,x_i))$ by $\epsilon n \log(en)$. Also since $p < 1/2$ (for large enough $n$ and $m \le n$), $1/(1-p) < 1+2p$. So we derive from ~\eqref{eq:source-coding} that
\begin{align}
\label{eqn:eqn3}
2 d(\med,x) \log n &+ \frac{1}{1+\delta}\sum_{x_i \in S} d(x,x_i)\log(2en/\epsilon) + \epsilon n \log (en) + 120 m \log^2 n \nonumber \\ 
& + 4 \epsilon^2 nm (\log (30/\epsilon) + \log \log n + 1) +1+2p \ge \epsilon nm \log(n/\epsilon).
\end{align}

Further note that for $\alpha=\frac{3}{\log(n/\epsilon)}$, $\log \Big(\frac{2en}{\epsilon}\Big)^{1/(1+\alpha)} < \log (n/\epsilon)$. Let $1+\delta=(1+\beta)(1+\alpha)$ for some $\beta>0$. Hence we can rewrite Equation~\ref{eqn:eqn3} as
\begin{align*}
\frac{2d(\med,x)}{n} &+ \frac{1}{1+\beta} (1+\frac{1}{\sqrt{\log n}})\epsilon m + \epsilon + \frac{m \log^2 n}{n} \\ &+ \frac{4\epsilon^2 m (\log (30/\epsilon) + \log \log n + 1)}{\log(n/\epsilon)} + \frac{1+2p}{n \log(n/\epsilon)} \ge \epsilon m.
\end{align*}
We can simplify the above as
\begin{align}
\frac{2d(\med,x)}{\epsilon n m} + \frac{1}{1+\beta} (1+\frac{1}{\sqrt{\log n}}) + \frac{1}{m} + f(n) \ge 1
\end{align}
for some real-valued function $f(\cdot)$ such that $f(n) \to 0$ as $n \to \infty$. Furthermore, by averaging argument we know that there exists $x_i$ such that $d(\med,x_i) \le d(x,x_i)$. So by the triangle inequality $d(\med,x) \le 2d(x,x_i) \le 2(1+\frac{1}{\sqrt{\log n}})\epsilon n$. As a consequence we get that $\beta \le 19/m$ for large enough $n$ and $m \ge 20$. This leads to a contradiction for large enough $n$ since $\delta =\frac{20}{m}+ \frac{3}{\log(n/\epsilon)}$. This concludes the proof of Theorem~\ref{thm:approx-median-hidden}.

Now it only remains to prove Claim~\ref{clm:decode-bounded}.
\begin{proof}[Proof of Claim~\ref{clm:decode-bounded}]
Let $r \in [n]$ be such that $x(r)=b_i(a)$, and $J=\{r-\frac{15}{\epsilon}\log n, \cdots, r+\frac{15}{\epsilon}\log n\}$. Note, a symbol $c \in x(J)$ is not is $x_i(J_i(a))$ only if either $c \in \Sigma_i$ and it moved out of this interval $J$, or there is some other symbol which is moved into this interval $J$ and as a result $c$ falls outside $\frac{30}{\epsilon} \log n$ sized interval. So, 
\begin{align*}
|x_i(J_i(a))\cap x(J)| &\ge |x(J)| - |x(J)\cap \Sigma_i| - |x(J) \cap \Sigma_i^b|\\
&\ge |x(J)|- 4 \epsilon |J|\\
&=(1-4 \epsilon) |J|\\
&= (1-4 \epsilon)\frac{30}{\epsilon} \log n
\end{align*}
where the second inequality holds with probability at least $1-2n^{-3}$ by Observation~\ref{obs:interval-intersection} and Observation~\ref{obs:interval-intersection-b}.

On the other hand for any interval $I \subseteq [n]$ on the left (or right) of $r$ (not including $r$),
\begin{align*}
|x_i(J_i(a))\cap x(I)| &\le \frac{15}{\epsilon} \log n + |x(I) \cap \Sigma_i|\\
&\le (\frac{1}{2}+ 2\epsilon) \frac{30}{\epsilon} \log n
\end{align*}
where the last inequality holds with probability at least $1-n^{-3}$ by Observation~\ref{obs:interval-intersection}. Now for any $\epsilon <1/12$, $1-4\epsilon > 1/2 +2 \epsilon$, and thus $x(R_i(a))$, which has maximum intersection with $x_i(J_i(a))$, must contain the symbol $x(r)=b_i(a)$.
\end{proof}

\subsection{Finding the Hidden Permutation}
In the last section, we have seen that to find an approximate median of a set $S$ drawn from $S(x,\epsilon,m)$, it suffices to find the permutation $x$. So from now on, we will focus only on finding $x$ (approximately).

\subsubsection*{When $m$ is large}
Apparently the task of finding the unknown $x$ becomes much easier when $m\ge \Omega(\log n)$. 
\begin{lemma}
\label{lem:avg-large-set}
For any $\epsilon \in (0,1/16)$, a large enough $n \in \mathbb{N}$, a permutation $x \in \sym_n$, and any $m \ge 32\log n$, let $S$ be drawn from $S(x,\epsilon,m)$. There is an $O(n \log^2 n)$ time algorithm that given $S$ outputs $x$ with probability at least $1-1/n$.
\end{lemma}
Note, running time of the algorithm is independent of $m$. The reason is that our algorithm will take an arbitrary $\Theta(\log n)$-sized subset of $S$ and compute $x$.
\begin{proof}
Finding $x$ is nothing but sorting the numbers in $[n]$ according to the order specified by $x$. Before proceeding further, let us introduce a notation that we will use henceforth. For any two distinct symbols $a, b \in [n]$ if $a$ appears before $b$ in $x$ we use the notation $a <_{x} b$. Below we describe our algorithm.

Without loss of generality, assume that $m=32\log n$; otherwise, take an arbitrary subset $S' \subseteq S$ of size $32\log n$ and perform our algorithm with $S'$ instead of $S$. To sort the symbols according to the ordering of $x$, we use the Mergesort \footnote{One may take any comparison-based sorting algorithm instead of the Mergesort; the running time will change accordingly.} with additional query access to the set $S$. While performing the Mergesort whenever two elements $a,b\in [n]$ will be compared to check whether $a <_{x} b$, we will use the following query algorithm.

\textbf{Query algorithm $(a,b)$:}
Compare $a,b$ in all $x_i \in S$. If at least in $m/2$ many $x_i$'s $a$ appears before $b$, then return $a <_{x} b$; else return $b <_{x} a$.

It follows from the time complexity of the Mergesort that the algorithm will make at most $O(n \log n)$ queries to our query algorithm. Each such query takes $O(m)$ time. So the total running time of our algorithm is $O(n \log^2 n)$, since by our assumption $m=32 \log n$.

Now it only remains to prove the correctness of our algorithm. For each $a \in [n]$ let $B_a=\{x_i \in S  :  a \in \Sigma_i\}$. Take a parameter $\delta=\frac{1}{4\epsilon} - 2$. We call a symbol $a \in [n]$ \emph{bad} if $|B_a| \ge (1+\delta)\epsilon m$. (Note, here the definition of a bad symbol is similar to that used in Section~\ref{sec:small-regime}. The only difference is that here our "unknown reference" is $x$ instead of a median string $\med$.) Consider any symbol $a \in [n]$. Then $\mathbb{E}[|B_a|] = \epsilon m$. Since $\Sigma_i$'s are generated independently of each other, by Chernoff bound
$$\Pr[a\text{ is bad}] \le e^{-\frac{\delta^2 \epsilon m}{2+\delta}}.$$
Then by a union bound over all symbols,
$$\Pr[\text{None of symbols is bad}] \ge 1- n e^{-\frac{\delta^2 \epsilon m}{2+\delta}} \ge 1-1/n$$
where the last inequality holds for $\epsilon < 1/16$, $\delta=\frac{1}{4\epsilon} - 2$ and $m=32 \log n$.

Observe, for any two distinct symbols $a,b \in [n]$ if $a <_x b$ and none of them is bad, then the number of $x_i$'s in $S$ in which in $a$ appears before $b$ is at least $(1 - 2(1+\delta)\epsilon) m > m/2$ for $\delta=\frac{1}{4\epsilon} - 2$. Thus our query algorithm always outputs a correct order among two symbols (given none of them is bad). The correctness now follows from the correctness of the Mergesort.
\end{proof}

\subsubsection*{When $m$ is small}
\begin{lemma}
\label{lem:avg-small-set}
For any $\epsilon \in (0,1/40)$, a large enough $n \in \mathbb{N}$, a permutation $x \in \sym_n$, and any $m$, let $S$ be drawn from $S(x,\epsilon,m)$. There is a (deterministic) algorithm that given $S$, outputs a permutation $\tilde{x} \in \sym_n$ such that $d(x,\tilde{x}) \le \frac{5}{3}(e^{-m/40} + 2\sqrt{\log n/n})n$ in time $O(n^3 + m n^2)$ with probability at least $1-1/n$.
\end{lemma}
\begin{proof}
Before describing the algorithm let us introduce a few notations to be used in this proof. For each $a \in [n]$ let $B_a=\{x_i \in S  :  a \in \Sigma_i\}$. Take a parameter $\delta=\frac{1}{10\epsilon} - 2$. We call a symbol $a \in [n]$ \emph{bad} if $|B_a| \ge \alpha |S|=\alpha m$, where $\alpha=(1+\delta)\epsilon$. (Note, here the definition of a bad symbol is similar to that used in Section~\ref{sec:small-regime}. The only difference is that here our "unknown reference" is $x$ instead of a median string $\med$.) Let 
$$G =\{a \in [n]  :  a\text{ is not bad}\},$$
and $\overline{G}=[n]\setminus G$.

Now we run the procedure {\BRO} (described in Section~\ref{sec:worst-polytime}) with $S$ and $\alpha$ as input. Next we show that this procedure will return a $\tilde{x} \in \sym_n$ with the desired distance bound from $x$.

 Consider any symbol $a\in [n]$. Then $\mathbb{E}[|B_a|] = \epsilon m$. Since $\Sigma_i$'s are generated independently of each other, by Chernoff bound
$$\Pr[a\text{ is bad}] \le e^{-\frac{\delta^2 \epsilon m}{2+\delta}}.$$
Let $p=e^{-\frac{\delta^2 \epsilon m}{2+\delta}} \le e^{-m/40}$ for any $\epsilon \in (0,1/40)$. So $\mathbb{E}[|G|] \ge (1-p)n$. Since a symbols is bad independent of any other symbol being bad, by Chernoff bound for any $\delta' \in (0,1)$,
\begin{align}
\label{eq:double-chernoff}
\Pr[|G| \ge (1-\delta')\mathbb{E}[|G|]] & \ge 1 - e^{-\frac{\delta'^2 \mathbb{E}[|G|]}{2}} \nonumber\\
& \ge 1 - e^{-\frac{\delta'^2 (1-p)n}{2}}.
\end{align}
Note, by our choice of parameter $\delta$ for any $\epsilon \in (0,1/40)$, $\alpha \in (0, 1/10)$ . So by an argument exactly the same as that used in the proof of Lemma~\ref{case1}, we get that 
\begin{align*}
d(x,\tilde{x}) & \le \frac{1}{1-4 \alpha}|\overline{G}| \nonumber\\
& \le \frac{1}{1-4\alpha}(p + \delta'(1-p))n \nonumber\\
& \le \frac{5}{3} (e^{-m/40} + \delta') n &\text{since }\alpha < 1/10
\end{align*}
where the second inequality holds with probability at least $1 - e^{-\frac{\delta'^2 (1-p)n}{2}}$ by ~\eqref{eq:double-chernoff}. Now to finish the proof set $\delta' = 2\sqrt{\log n/n}$.
\end{proof}

\paragraph*{Proof of Theorem~\ref{thm:avg-case-poly}. }Now we are ready to finish the proof of Theorem~\ref{thm:avg-case-poly}. For $m \ge 32 \log n$, Theorem~\ref{thm:approx-median-hidden} together with Lemma~\ref{lem:avg-large-set} shows that in time $O(n \log^2 n)$ we can find a $(1+\delta)$-approximate median of $S$ drawn from $S(x,\epsilon,m)$, for $\delta=\frac{20}{m}+\frac{3}{\log(n/\epsilon)}$ with probability at least $1-mn^{-1.5}$.

For any $m < 32 \log n$ by Lemma~\ref{lem:avg-small-set} we get a $\tilde{x}$ such that $d(x,\tilde{x}) \le \frac{5}{3} (e^{-m/40} + 2 \sqrt{\log n/n}) n$ with probability at least $1-1/n$. Let $\gamma = e^{-m/40} + 2 \sqrt{\log n/n}$.

\begin{align*}
\obj(S,\tilde{x}) &= \sum_{x_i \in S}d(x_i,\tilde{x})\\
& \le \sum_{x_i \in S}d(x_i,x) + md(x,\tilde{x})& \text{by the triangle inequality}\\
& \le \obj(S,x) + \frac{5}{3} \gamma nm & \text{by Lemma~\ref{lem:avg-small-set}}\\
& \le \Big(1+\frac{\gamma}{\epsilon}\Big) \obj(S,x)&\text{by Observation~\ref{obs:size-move} w.p. at least }1-1/m\\
& \le \Big(1+\frac{\gamma}{\epsilon}\Big)\Big(1+\frac{20}{m}+\frac{3}{\log(n/\epsilon)}\Big)\opt(S)&\text{by Theorem~\ref{thm:approx-median-hidden}}\\
& \le \Big(1 + \frac{20}{m}+\frac{3}{\log(n/\epsilon)} + \frac{2 e ^{-m/40}}{\epsilon}  \Big) \opt(S)& \text{by replacing the value of } \gamma.
\end{align*}
This concludes the proof of Theorem~\ref{thm:avg-case-poly}.

\section{Exact Median for Three Permutations}
\label{sec:worst-exptime}
In this section we provide an algorithm that given three permutation $x_1,x_2,x_3\in \mathcal{S}_n$, finds their exact median permutation $\tilde{x}$ in time $O(mn^3)$. For a symbol $s\in [n]$, let $freq_{x}(s)$ represent the number of times $s$ appears in $x$ and $x(i,\dots,j)$ represent the substring of $x$ staring at index $i$ and ending at index $j$.

\begin{theorem}
	\label{thm:threeulam}
There is an algorithm that given three permutations $x_1,x_2,x_3\in {\mathcal{S}}_n$, computes a permutation $\tilde{x}$ such that 
$\tilde{x}=\argmin_{y\in \mathcal{S}_n}\sum_{i\in [3]}d(y,x_i)$, in time $O(n^4)$.
\end{theorem}

Before going to the details of the algorithm, we make the following observations. Let $x,y\in \mathcal{S}_n$. Their Ulam distance denoted by $d(x,y)$ is the minimum number of character move required to convert one string to the other. Next, we define the edit distance of two given strings $x,y$ of length $n$ (not necessarily permutation) denoted by $\Delta(x,y)$ to be the minimum number of character insertions and deletions required to convert one string to the other (note in general we allow substitution as well, but for our analysis purpose we omit this). Notice as every move operation can be represented by one deletion followed by an insertion operation; we can make the following trivial claim. 

\begin{claim}
	\label{obs:equal}
	Given two strings $x,y\in \mathcal{S}_n$, $d(x,y)=\Delta(x,y)/2$
\end{claim}

\begin{lemma}
	\label{lem:ulamedit}
	Let $x_1,\dots, x_m\in \mathcal{S}_n$. Let $\overline{x}$ and $\tilde{x}$  be the two permutations such that\\ $\overline{x}=\argmin_{y\in \mathcal{S}_n}\sum_{i\in [m]}d(y,x_i)$ and  $\tilde{x}=\argmin_{z\in \mathcal{S}_n}\sum_{i\in [m]}\Delta(z,x_i)$. Then, $\sum_{i\in [m]}d(\tilde{x},x_i)=\sum_{i\in [m]}d(\overline{x},x_i)$. 
\end{lemma}

\begin{proof}
	As $\overline{x}=\argmin_{y\in \mathcal{S}_n}\sum_{i\in [m]}d(y,x_i)$, $\sum_{i\in [m]}d(\tilde{x},x_i)\ge\sum_{i\in [m]}d(\overline{x},x_i)$. Next we show $\sum_{i\in [m]}d(\tilde{x},x_i)\le\sum_{i\in [m]}d(\overline{x},x_i)$. For contradiction assume $\sum_{i\in [m]}d(\tilde{x},x_i)>\sum_{i\in [m]}d(\overline{x},x_i)$. Then, $2\sum_{i\in [m]}d(\tilde{x},x_i)> 2\sum_{i\in [m]}d(\overline{x},x_i)$ and by Claim~\ref{obs:equal}, $\sum_{i\in [m]}\Delta(\tilde{x},x_i)> \sum_{i\in [m]}\Delta(\overline{x},x_i)$. Therefore we get a contradiction as $\tilde{x}=\argmin_{z\in \mathcal{S}_n}\sum_{i\in [m]}\Delta(z,x_i)$.
\end{proof}

In the rest of the section our objective will be to design an algorithm that given $x_1,x_2, x_3\in \mathcal{S}_n$, generates a permutation $\tilde{x}$ such that $\tilde{x}=\argmin_{z\in \mathcal{S}_n}\sum_{i\in [3]}\Delta(z,x_i)$.	

\begin{theorem}
	\label{thm:threeedit}
	There is an algorithm that given three permutations $x_1,x_2,x_3\in {\mathcal{S}}_n$, computes a permutation $\tilde{x}$ such that 
	$\tilde{x}=\argmin_{y\in \mathcal{S}_n}\sum_{i\in [3]}\Delta(y,x_i)$, in time $O(n^4)$.
\end{theorem}
	
\noindent	
Note that Theorem~\ref{thm:threeulam} directly follows from Theorem~\ref{thm:threeedit} and Lemma~\ref{lem:ulamedit}.

\vspace{2mm}
\noindent	
The non trivial part of our algorithm lies in the fact that if we just run the conventional dynamic program to find $\tilde{x}$, the string output by the algorithm may not be a permutation. In fact its length can be different from $n$ as well. Therefore we first generate a $n$-length string $x'$ (not necessarily a permutation) such that $x'=\argmin_{y\in [n]^n}\sum_{i\in[3]}\Delta(y,x_i)$
using dynamic program. Next we postprocess $x'$ to generate a permutation $\tilde{x}$ over $[n]$, by removing all multiple occurrences of a same symbol and by inserting all the missing symbols. We then argue that $\sum_{i\in[3]} \Delta(\tilde{x},x_i)=\sum_{i\in [3]} \Delta(x',x_i)$. 
Theorem~\ref{thm:threeedit} follows from the fact that for any $x''=\argmin_{y\in \mathcal{S}_n}\sum_{i\in[3]}\Delta(y,x_i)$, $\sum_{i\in[3]} \Delta(x',x_i)\le \sum_{i\in [3]} \Delta(x'',x_i)$.

\vspace{2mm}
\noindent
We now describe our algorithm. Broadly the algorithm has two steps.

\noindent
\textbf{STEP 1:} In step 1 given three strings 
 $x_1,x_2,x_3\in {\mathcal{S}}_n$, we use the following dynamic program to find a string $x'$ such that $x'=\argmin_{y\in [n]^n}\sum_{i\in[3]}\Delta(y,x_i)$. The dynamic program uses the following operations: 
 For any symbol $a\in[n]$, $(i)$ $(\phi,a)$ represents insertion of symbol $a$ $(ii)$ $(a,\phi)$ represents deletion of symbol $a$ and for any two symbols $a,b\in[n]$, $(iii)$ $(a,b)$ represents deletion of symbol $a$ followed by insertion of symbol $b$. Let $w(\phi,a), w(a,\phi)$ denotes the cost the insertion and deletion operation respectively. 
  For our purpose $w(a,\phi)=w(\phi,a)=1$ and $w(a,b)=w(a,\phi)+w(\phi,b)=2$. We define the dynamic program matrix $D$ as follows: $D_{i,j,k,\ell}$ represents $\min_{s\in[n]^\ell}(\Delta(s, x_1(1,\dots,i))+\Delta(s, x_2(1,\dots,j))+\Delta(s, x_3(1,\dots,k)))$. We can represent $D_{i,j,k,\ell}$ recursively as follows:
 
 \begin{equation}
 D_{i,j,k,\ell}= min
 \begin{cases}
 D_{i-1,j-1,k-1,\ell-1}+\min_{a\in[n]}[w(x_{1_i},a)+w(x_{2_j},a)+w(x_{3_k},a)]\\
 D_{i,j-1,k-1,\ell-1}+\min_{a\in[n]} [w(\phi,a)+w(x_{2_j},a)+w(x_{3_k},a)]\\
 D_{i-1,j,k-1,\ell-1}+ \dots\\%\min_{d\in[n]} [w(x_{1_i},\phi)+w(\phi,d)+w(x_{3_k},\phi)]\\
 D_{i-1,j-1,k,\ell-1}+\dots\\%\min_{d\in[n]} [w(x_{1_i},\phi)+w(x_{2_j},d)+w(\phi,d)]\\
 D_{i,j,k-1,\ell-1}+\min_{a\in[n]}[2w(\phi,a)+w(x_{3_k},a)]\\
 D_{i,j-1,k,\ell-1}+\dots\\%\min_{d\in[n]}[2w(\phi,d)+w(x_{2_j},d)]\\
 D_{i-1,j,k,\ell-1}+\dots\\%\min_{d\in[n]}[2w(\phi,d)+w(x_{1_i},d)]\\
 D_{i-1,j-1,k-1,\ell}+w(x_{1_i},\phi)+w(x_{2_j},\phi)+w(x_{3_k},\phi)\\
 D_{i,j-1,k-1,\ell}+w(x_{2_j},\phi)+w(x_{3_k},\phi)\\
 D_{i-1,j,k-1,\ell}+\dots\\%w(x_{1_i},\phi)+w(x_{3_k},\phi)\\
 D_{i-1,j-1,k,\ell}+\dots\\%w(x_{1_i},\phi)+w(x_{2_j},\phi)\\
 D_{i,j,k-1,\ell}+w(x_{3_k},\phi)\\
 D_{i,j-1,k,\ell}+\dots\\%w(x_{2_j},\phi)\\
 D_{i-1,j,k,\ell}+\dots \\%w(x_{1_i},\phi)\\
 
 \end{cases}
 \end{equation}
 
 \noindent
 Once $D_{n,n,n,n}$ is computed, we can simply back track to find one $n$ length optimal string $x'$.

\vspace{2mm}
\noindent
\textbf{STEP 2:} In step 2 given string $x'$ we generate a string $\tilde{x}$ such that $\tilde{x}=\argmin_{y\in \mathcal{S}_n}\sum_{i\in[3]}\Delta(y,x_i)$. If $x'$ is already a permutation set $\tilde{x}=x'$. Otherwise perform the following two operations and output the string generated.

\begin{enumerate}[i]
	\item Fix an optimal alignment $A_i$ between each $x_i$ an $x'$. For each symbol, $s\in [n]$, that appears more than once in $x'$, delete all but the one that maximizes the matches in $A_i$s (If there are more than one occurrences that maximize the matches keep an arbitrary one and delete the rest). Let the new string be $x''$. For each $x_i$, let $A'_i$ be the alignment obtained from $A_i$ by keeping only those matches whose corresponding characters are not deleted in $x''$.
	
	\item Consider each $\ell\in[n]$ in increasing order and if $x_1(\ell)=s$ does not appear in $x''$(and therefore in $x'$) do the following: let $j$ be the largest index $<\ell$, such that $x_1(j)$ has a match (with $x''(k)$) in $A'_1$. Insert $s$ after $x''(k)$ in $x''$. Accordingly modify $A'_1$ by adding a match between $x_1(\ell)$ and $x''(k+1)$. Let the final string be $\tilde{x}$.
\end{enumerate}

\begin{lemma}
	\label{lem:delins}
	$\sum_{i\in[3]} \Delta(\tilde{x},x_i)=\sum_{i\in[3]} \Delta(x',x_i)$. 
\end{lemma}

\begin{proof}

%To prove this we show for three strings $\sum_{i} d(M',x_i)=\sum_{i} d(M,x_i)$. 
Let $R\subseteq [n]$ be the set of symbols that appear more than once in $x'$ and $M\subseteq [n]$ be the set of symbols that do not appear in $x'$.  Let $\mathcal{C}= \sum_{s\in S_1} ({freq}_{x'}(s)-1)$.
Now here each repeated occurrence of a symbol can contribute a match to at most one $A_i$. Hence for each of such repeated occurrences, in the rest of the $2$ strings we need to pay for one deletion and one insertion. Moreover for each missing symbol, we pay one deletion for each of the $3$ strings. Therefore, $\sum_{i} \Delta(x',x_i)\ge 4\mathcal{C}+3\mathcal{C}$. Hence, $\mathcal{C}\le (\sum_{i} \Delta(x',x_i))/7$.

Now at Step 1 for each deletion we loose at most one match from some $A_i$ and for this we need to pay for one extra deletion. But in each of the rest $2$ alignments we gain one deletion. Hence, $\sum_{i\in[3]} \Delta(x'',x_i)=\sum_{i\in[3]} \Delta(x',x_i)-|\mathcal{C}|$.
On the other hand each insertion at Step 2 adds an extra match to $A'_1$, while not loosing any matches from $A'_2$ and $A'_3$. Hence, after step 2 for each $s\in M$ we gain one deletion from $x_1$ and pay $2$ insertions for inserting $s$ in both $x_2$ and $x_3$. 

\begin{align*}
\sum_{i} \Delta(\tilde{x},x_i)&\le \sum_{i\in[3]} \Delta(x'',x_i)+2\mathcal{C}-\mathcal{C}\\
&= \sum_{i\in[3]} \Delta(x'',x_i)+\mathcal{C}\\
&= \sum_{i\in[3]} \Delta(x',x_i)-\mathcal{C}+\mathcal{C}\\
&=\sum_{i} \Delta(x',x_i)\\
\end{align*}
\end{proof}

\paragraph*{Running time analysis. } In step $1$ string $x'$ can be computed using dynamic program in time $O(n^4)$.  In step $2.i$ first we compute optimal alignments $A_i$ for each $i\in[n]$ in time $O(n^2)$. Next for each symbol in $[n]$, we delete their repetitive copies in time $O(n^2)$. Lastly in step $2.ii$, insertion of each missing symbol can again be performed in time $O(n^2)$. Hence the total time required is $O(n^4)$. 

\begin{proof}[Proof of Theorem~\ref{thm:threeedit}]
The proof directly follows from the fact that the string $x'$ generated by Step $1$ of the algorithm ensures that $x'=\argmin_{y\in [n]^n}\sum_{i\in[3]}{\Delta}(y,x_i)$, therefore for any $x''=\argmin_{y\in \mathcal{S}_n}\sum_{i\in[3]}\Delta(y,x_i)$, $\sum_{i\in[3]} \Delta(x',x_i)\le \sum_{i\in [3]} {\Delta}(x'',x_i)$ and Lemma~\ref{lem:delins}.
\end{proof}

\paragraph{Generalization for $m$ permutations.}

In this section we provide a generalization of the above result for $m$ permutations. Notice for $m$ strings also, the normal dynamic program may output a string that is not a permutation. Therefore similar to above we use some extra post process step. The only difference is that here using post processing on an optimal $n$-length median string we generate a permutation that is $1.5$ approximation of the optimal median permutation. Given $m$ permutations $x_1,\dots,x_m\in \mathcal{S}_n$, let $\opt=\min_{y\in\mathcal{S}_n}\sum_{i\in[m]}d(y,x_i)$ and $\opt_{\Delta}=\min_{y\in\mathcal{S}_n}\sum_{i\in[m]}\Delta(y,x_i)$ We show the following.  

\begin{theorem}
	\label{thm:threeulamg}
	There is an algorithm that given $m$ permutations $x_1,\dots,x_m\in {\mathcal{S}}_n$, computes a permutation $\tilde{x}\in {\mathcal{S}}_n$ such that 
	$\sum_{i\in [m]}d(\tilde{x},x_i)\le 1.5\opt$, in time $O(2^{m+1}n^{m+1})$.
\end{theorem}

Similar to the above, instead of proving this theorem directly, we prove the following. 

\begin{theorem}
\label{thm:threeeditg}
	There is an algorithm that given $m$ permutations $x_1,\dots,x_m\in {\mathcal{S}}_n$, computes a permutation $\tilde{x}\in {\mathcal{S}}_n$ such that  
	$\sum_{i\in [m]}\Delta(\tilde{x},x_i)\le 1.5\opt_{\Delta}$, in time $O(2^{m+1}n^{m+1})$.
\end{theorem}

Note Theorem~\ref{thm:threeulamg} can be directly implied from Theorem~\ref{thm:threeeditg} and Lemma~\ref{lem:ulamedit}.

\vspace{2mm}
\noindent
Same as above, the algorithm has two steps.

\noindent
\textbf{STEP 1:} In step 1 given $m$ permutations $x_1,\dots,x_m\in {\mathcal{S}}_n$, we use the dynamic program to find a string $x'$ such that $x'=\argmin_{y\in [n]^n}\sum_{i\in[m]}\Delta(y,x_i)$. 
The dynamic program used here is a direct generalization of the dynamic program designed for three strings case.

\vspace{2mm}
\noindent
\textbf{STEP 2:} In step 2 given string $x'$ we generate a string $\tilde{x}$ such that $\sum_{i\in [m]}\Delta(\tilde{x},x_i)\le 1.5\sum_{i\in [m]}\Delta(x',x_i)$. If $x'$ is already a permutation set $\tilde{x}=x'$. Otherwise postprocess string $x'$, first by deleting all but one occurrence (keep the one that maximizes the matches in $A_i$s ) of each symbol in $x'$. Next similar to the three string case, insert all the symbols missing from $x'$ while ensuring each insertion creates a new match in $A'_1$ while keeping all the previous matches intact. Let the final string be $\tilde{x}$.

\iffalse
\begin{enumerate}[i]
	\item Fix an optimal alignment $A_i$ between each $x_i$ an $x'$. For each symbol, $s\in [n]$, that appears more than once in $x'$, delete all but one that maximizes the matches in $A_i$s (If there are more than one occurrences that maximize the matches keep an arbitrary one and delete the rest). Let the new string be $x''$. For each $x_i$, let ${A'}_i$ be the alignment obtained from $A_i$ by keeping only those matches whose corresponding character is not deleted in $x''$.
	
	\item Consider all $\ell\in[n]$ in increasing order and if $x_1[\ell]=s$ does not appear in $x''$(and therefore in $x'$) do the following: let $j$ be the largest index $<\ell$, such that $x_1[j]$ has a match (with $x''[k]$) in ${A'}_1$. Insert $s$ after $x''[k]$ in $x''$. Accordingly modify ${A'}_1$ by adding a match between $x_1[\ell]$ and $x''[k+1]$. Let the final string be $\tilde{x}$.
\end{enumerate}
\fi

\begin{lemma}
	\label{lem:delinsg}
	$\sum_{i\in [m]} \Delta(\tilde{x},x_i)=3/2(\sum_{i\in[m]} \Delta(x',x_i))$. 
\end{lemma}

\begin{proof}
	Define $R, M, \mathcal{C}$ in a similar way as for the three string case. 
	Note,  at Step 1 we delete exactly $\mathcal{C}$ characters from $x'$ and at Step 2 we insert exactly $\mathcal{C}$ missing symbols in $x'$. %$2(\mathcal{C}-|S|)<\sum_{i\in[3]} d(M,x_i)$. 
	Now each repeated occurrence of a symbol can contribute a match to at most $m/2$ different $A_i$s. Hence for each of such repeated occurrences, in the rest of the $\ge m/2$ strings we need to pay for one deletion and one insertion. Moreover for each missing symbol, we pay one deletion for each of the $m$ strings. Therefore, $\sum_{i\in [m]} \Delta(x',x_i)\ge 2(m/2)\mathcal{C}+m\mathcal{C}$. Hence, $\mathcal{C}\le (\sum_{i\in [m]} \Delta(x',x_i))/2m$.
	Now at Step 1 for each deletion we loose at most $m/2$ matches from $m/2$ different $A_i$s and for each of them we need to pay for one extra deletion. But in each of the rest $\ge m/2$ alignments we gain one deletion. Hence, $\sum_{i\in [m]} \Delta(x'',x_i)\le \sum_{i\in [m]} \Delta(x',x_i)$.
	On the other hand each insertion at Step 2 adds an extra match to $A'_1$, while not loosing any matches from $A'_2,\dots,A'_m$. Hence, after step 2 we need to pay for inserting each $s\in M$ in each of $x_2,\dots,x_m$. Therefore, 
	
	\begin{align*}
	\sum_{i\in [m]} \Delta(\tilde{x},x_i)&\le \sum_{i\in [m]} \Delta(x'',x_i)+\mathcal{C}(m-1)\\
	&\sum_{i\in [m]} \Delta(x',x_i)+\mathcal{C}(m-1)\\
    &\le 3/2(\sum_{i\in[m]} \Delta(x',x_i)) \text{   (As, $\mathcal{C}\le (\sum_{i\in [m]} \Delta(x',x_i)/2m$)}
	\end{align*}
	
\end{proof}

Clearly, from the construction of the dynamic program and step 2, the algorithm has a running time of $O(2^{m+1}n^{m+1})$. Moreover, the proof of Theorem~\ref{thm:threeeditg} directly follows from the definition of $x'$ and Lemma~\ref{lem:delinsg}.

\paragraph*{Acknowledgements.} Diptarka would like to thank Djordje Jevdjic for some useful discussion on~\cite{RMRAJY17}.

%%%%%%%%%%%%%%%%%%%%%%%%%%%%%%%%%%%%%%%%%%%%%%%%%%%%%%%%%%%
\bibliographystyle{alphaurl}
\bibliography{MedianUlam} 

%%%%%%%%%%%%%%%%%%%%%%%%%%%%%%%%%%%%%%%%%%%%%%%%%%%%%%%%%%%

%\input{appendix}

{\tiny }
\end{document}